\documentclass[envcountsame]{llncs}
\usepackage{ifpdf}
\ifpdf
\usepackage{pdfsync}
\fi
\usepackage{a4wide}
\usepackage[tbtags]{amsmath}
\usepackage{amssymb}
\usepackage{amsfonts}
\usepackage{mathrsfs}
\usepackage{enumerate}
\usepackage{gastex}
\usepackage{graphics}
\usepackage{xcolor}
\usepackage{latexsym}
\usepackage{ifthen}
\usepackage{url}

\newcommand{\A}{\mathcal A}

\newcommand{\B}{\mathcal B}
\newcommand{\C}{\mathcal C}

\newcommand{\unfold}{\mathrm{unfold}}
\newcommand{\E}{\mathcal E}
\newcommand{\F}{\mathcal F}
\newcommand{\first}{\mathrm{first}}
\newcommand{\FO}{\mathsf{FO}}
\newcommand{\FOTh}{\mathsf{FOTh}}

\newcommand{\Good}{\mathrm{Good}}

\newcommand{\Img}{\mathrm{Img}}
\renewcommand{\L}{\mathcal L}
\newcommand{\lex}{\mathrm{lex}}
\newcommand{\Loop}{\mathrm{loop}}
\newcommand{\N}{\mathbb N}
\newcommand{\M}{\mathcal M}
\renewcommand{\L}{\mathcal L}
\newcommand{\K}{\mathcal K}
\renewcommand{\P}{{\mathbb P}}
\newcommand{\Q}{\mathbb Q}

\newcommand{\Run}{\mathrm{Run}}
\renewcommand{\S}{\mathcal S}
\newcommand{\Shuf}{\mathrm{Shuf}}
\newcommand{\T}{\mathcal T}
\newcommand{\Z}{\mathbb Z}

\begin{document}
\title{The Isomorphism Problem On Classes of Automatic Structures}
\author{Dietrich Kuske\inst{1} \and Jiamou Liu\inst{2} \and  Markus Lohrey\inst{2,}
\thanks{The second and third author are supported by the DFG  research project GELO.}
\institute{
Centre national de la recherche scientifique (CNRS)
  and Laboratoire Bordelais de Recherche en Informatique
  (LaBRI), Bordeaux, France
\and
Universit\"at
Leipzig, Institut f\"ur Informatik, Germany
\\
\email{kuske@labri.fr, liujiamou@gmail.com,
lohrey@informatik.uni-leipzig.de}}}

\maketitle
\pagestyle{plain}

\begin{abstract}
Automatic structures are finitely presented structures
where the universe and all relations can be recognized by finite automata.
It is known that
the isomorphism problem for automatic structures is complete for $\Sigma^1_1$;
the first existential level of the analytical hierarchy.
Several new results on isomorphism problems for automatic structures are shown in this paper:
(i) The isomorphism problem for automatic equivalence relations is complete for $\Pi^0_1$
(first universal level of the arithmetical hierarchy).
(ii) The isomorphism problem for automatic trees of height $n \geq 2$ is
 $\Pi^0_{2n-3}$-complete.
(iii) The isomorphism problem for automatic linear orders is not arithmetical.
This solves some open questions of Khoussainov, Rubin, and Stephan. 
\end{abstract}

\section{Introduction}

The idea of an automatic structure goes back to B\"uchi and Elgot who
used finite automata to decide, e.g., Presburger
arithmetic~\cite{Elg61}. Automaton decidable theories~\cite{Hod82}
and automatic groups~\cite{EpsCHLPT92} are similar concepts. A
systematic study was initiated by Khoussainov and Nerode~\cite{KhoN95}
who also coined the name ``{\em automatic structure}''.
In essence, a structure is automatic if the elements of the universe can be
represented as strings from a regular language and every relation of the structure
can be recognized by a finite state automaton with several heads that
proceed synchronously.  Automatic structures received
increasing interest over the last
years~\cite{BarKR08,BluG04,IsKhRu02,KhoNRS07,KhoRS05,KusLo09JSL,Rub08}.
One of the main motivations for investigating automatic structures is that
their first-order theories can be decided uniformly (i.e., the input
is an automatic presentation and a first-order sentence).

Automatic structures form a subclass of recursive (or computable) structures.
A structure is recursive, if its domain as well as all relations are recursive
sets of finite words (or naturals). A well-studied problem for
recursive structures is the isomorphism problem, where it is asked whether
two given recursive structures over the same signature
(encoded by Turing-machines for the domain
and all relations) are isomorphic. It is well known that the isomorphism
problem for recursive structures is complete for the first level of the analytical hierarchy
$\Sigma^1_1$. In fact, $\Sigma^1_1$-completeness holds for many subclasses
of recursive structures, e.g., for linear orders, trees, undirected graphs, Boolean
algebras, Abelian $p$-groups, see \cite{CaKni06,GonKn02}.
$\Sigma_1^1$-completeness of the isomorphism problem for a class of recursive
structures implies non-existence of a good classification
(in the sense of \cite{CaKni06}) for that class \cite{CaKni06}.

In \cite{KhoNRS07}, it was shown that also for automatic structures the
isomorphism problem is $\Sigma^1_1$-complete. By a direct interpretation,
it follows that for the following classes the isomorphism problem is
still $\Sigma^1_1$-complete \cite{Nie07}: automatic successor trees, automatic undirected
graphs, automatic commutative monoids, automatic partial orders, automatic
lattices of height 4, and automatic 1-ary functions.
On the other hand, the isomorphism problem is decidable for automatic ordinals
\cite{KhoRS05} and automatic Boolean algebras \cite{KhoNRS07}.
An intermediate class is the class of all locally-finite automatic graphs, for
which the isomorphism problem is complete for $\Pi^0_3$ (third level
of the arithmetical hierarchy\footnote{For background on the arithmetical 
hierarchy see, e.g., \cite{Rogers}.}) \cite{Rub04}.

For many interesting classes of automatic structures, the exact status of the isomorphism problem
is open. In the recent survey \cite{Rub08} it was asked for instance, whether
the isomorphism problem is decidable for automatic equivalence relations and
automatic linear orders.
For the latter class, this question was already asked in \cite{KhoRS05}.
In this paper, we answer these questions.
Our main results are:
\begin{itemize}
\item The isomorphism problem for  automatic equivalence relations
is $\Pi^0_1$-complete.
\item The isomorphism problem for automatic successor trees of finite height
$k \geq 2$ (where the height of a tree is the maximal number of edges along
a maximal path) is $\Pi^0_{2k-3}$-complete.
\item The isomorphism problem for automatic linear orders is
hard for every level of the arithmetical hierarchy.
\end{itemize}
Most hardness proofs for automatic structures, in particular the
$\Sigma^1_1$-hardness proof for the isomorphism problem of automatic
structures from \cite{KhoNRS07}, use transition graphs of
Turing-machines (these graphs are easily seen to be automatic). This
technique seems to fail for inherent reasons, when trying to prove our
new results. The reason is most obvious for equivalence relations and
linear orders. These structures are transitive but the transitive
closure of the transition graph of a Turing-machine cannot be
automatic in general (it's first-order theory is undecidable in
general). Hence, we have to use a new strategy. Our proofs are based
on the undecidability of Hilbert's $10^{th}$ problem. Recall that
Matiyasevich proved that every recursively enumerable set of natural
numbers is Diophantine \cite{Mat93}. This fact was used by Honkala to
show that it is undecidable whether the range of a rational power
series is $\mathbb{N}$ \cite{Hon06}.  Using a similar encoding, we
show that the isomorphism problem for automatic equivalence relations
is $\Pi^0_1$-complete. Next, we extend our technique in order to show
that the isomorphism problem for automatic successor trees of height
$k \geq 2$ is $\Pi^0_{2k-3}$-complete. In some sense, our result
for equivalence relations makes up the induction base $k=2$.
Finally, using a similar but
technically more involved reduction, we can show that the isomorphism
problem for automatic linear orders is hard for every level of the
arithmetical hierarchy. In fact, since our proof is uniform on the
levels in the arithmetical hierarchy, it follows that
the isomorphism problem for automatic linear orders is at least
as hard as true arithmetic (the first-order theory of
$(\mathbb{N};+,\times)$).  
At the moment it remains open whether the isomorphism problem for
automatic linear orders is $\Sigma^1_1$-complete.

\section{Preliminaries} \label{sec:prelim}

Let $\N_+ = \{1,2,3,\ldots\}$. Let $p(x_1,\ldots,x_n) \in \N[x_1,\ldots,x_n]$
be a polynomial with non-negative integer coefficients.
We define
$$
\Img_+(p) = \{ p(y_1,\ldots,y_n) \mid y_1,\ldots,y_n \in \N_+\}.
$$
If $p$ is not the zero-polynomial, then $\Img_+(p) \subseteq \N_+$.

Details on the arithmetical hierarchy can be found for instance
in \cite{Rogers}. With $\Sigma^0_n$ we denote the $n^{th}$
(existential) level of the arithmetical hierarchy; it is the class
of all subsets $A \subseteq \mathbb{N}$ such that
there exists a recursive predicate $P \subseteq \mathbb{N}^{n+1}$
with
$$
A = \{ a \in \mathbb{N}\mid \exists x_1 \forall x_2 \cdots Q x_n : (a,x_1,\ldots,x_n) \in P \},
$$
where $Q = \exists$ ($Q = \forall$) for $n$ odd (even).
The set of complements of $\Sigma^0_n$-sets is denoted
by $\Pi^0_n$. By fixing some effective encoding of strings
by natural numbers, we can talk about $\Sigma^0_n$-sets and 
$\Pi^0_n$-sets of strings over an arbitrary alphabet.
A typical example of a set, which does not belong
to the arithmetical hierarchy is {\em true arithmetic}, i.e., 
the first-order theory of $(\mathbb{N};+,\times)$, which we denote 
by $\FOTh(\mathbb{N};+,\times)$. 

We assume basic terminologies and notations in automata theory 
(see, for example, \cite{HoUl79}).
For a fixed alphabet $\Sigma$, a {\em non-deterministic finite automaton} is a
tuple $\A=(S, \Delta, I, F)$ where $S$ is the set of states,
$\Delta \subseteq S \times \Sigma \times S$
is the transition relation, $I\subseteq S$
is a set of initial states, and $F\subseteq S$ is the set of accepting states.
A {\em run} of $\A$ on a word $u = a_1 a_2 \cdots a_n$ ($a_1, a_2 \ldots,a_n \in \Sigma$)
is a word over $\Delta$ of the form $r = (q_0, a_1, q_1) (q_1, a_2, q_2) \cdots (q_{n-1}, a_n, q_n)$,
where $q_0 \in I$. If moreover $q_n \in F$, then  $r$ is an {\em accepting run} of $\A$ on $u$.
We will only apply these definitions in case $n > 0$, i.e., we will only speak of (accepting)
runs on non-empty words.

Given two automata $\A_1 = (S_1, \Delta_1, I_1, F_1)$ and $\A_2 = (S_1, \Delta_2, I_1, F_1)$ over the same alphabet 
$\Sigma$, we use $\A_1\uplus \A_2$ to denote the automaton obtained 
by taking the disjoint union of $\A_1$ and $\A_2$. Note that for any word $u \in \Sigma^+$, the number of 
accepting runs of $\A_1\uplus \A_2$ on $u$ is equal to the sum of the numbers of accepting runs of $\A_1$ and  $\A_2$ on $u$.
We use $\A_1\times \A_2$ to denote the Cartesian product of $\A_1$ and $\A_2$.
It is the automaton $(S_1 \times S_2, \Delta, I_1 \times I_2, F_1 \times F_2)$, where
$$
\Delta = \{ ((p_1,p_2),\sigma, (q_1,q_2)) \mid (p_1,\sigma,q_1) \in \Delta_1, (p_2,\sigma,q_2) \in \Delta_2\}.
$$
Then, clearly, the number of accepting runs of $\A_1 \times \A_2$ on a word 
$u \in L(A_1) \cap L(A_2)$ is the product of the numbers of accepting runs 
of $\A_1$ and  $\A_2$ on $u$.
In particular, if $\A_1$ is deterministic, then 
the number of accepting runs of 
$\A_1\times \A_2$ on $u\in L(\A_1) \cap L(A_2)$ is the same as the number of accepting runs of $\A_2$ on $u$.
In the following, if $\A$ is a non-deterministic automaton 
and $D$ is a regular language, we write
$D \uplus \A$ (resp. $D \cap \A$)
for the automaton $\A_D \uplus \A$ (resp. $\A_D \times \A$), where 
$\A_D$ is some deterministic automaton for the language $D$.

We use {\em synchronous $n$-tape automata} to recognize 
$n$-ary relations. Such automata have $n$ input tapes, each of which
contains one of the input words. The $n$ tapes are read in parallel until all input words are processed.
Formally, let $\Sigma_\diamond = \Sigma \cup \{\diamond\}$ where $\diamond
\notin \Sigma$. For words $w_1,w_2,\ldots,w_n \in \Sigma^{*}$, their {\em convolution} is a word
$w_1\otimes \cdots \otimes w_n \in (\Sigma_\diamond^n)^*$ with length $\max\{|w_1|,\ldots,|w_n|\}$, and the $k^{th}$ symbol of
$w_1\otimes\cdots \otimes w_n$ is $(\sigma_1,\ldots,\sigma_n)$ where
$\sigma_i$ is the $k^{th}$ symbol of $w_i$ if $k \leq |w_i|$, and
$\sigma_i=\diamond$ otherwise. An $n$-ary relation $R$ is
\emph{FA recognizable} if the set of all convolutions of tuples
$(w_1,\ldots,w_n)\in R$ is a regular
language.

A {\em relational structure} $\S$ consists of a {\em domain} $D$ and
atomic relations on the set~$D$. We will only consider structures with
countable domain. If $\S_1$ and $\S_2$ are two structures over the
same signature and with disjoint domains, then we write $\S_1 \uplus
\S_2$ for the union of the two structures. Hence, when writing $\S_1
\uplus \S_2$, we implicitly express that the domains of $\S_1$ and
$\S_2$ are disjoint. More generally, if $\{ \S_i \mid i \in I \}$ is a
class of pairwise disjoint structures over the same signature, then we
denote with $\uplus\{ \S_i \mid i \in I \}$ the union of these
structures. A structure $\S$ is called {\em automatic} over $\Sigma$
if its domain is a regular subset of $\Sigma^*$ and each of its atomic
relations is FA recognizable; any tuple~$\P$ of automata that accept
the domain and the relations of $\S$ is called an \emph{automatic
  presentation of $\S$}; in this case, we write $\S(\P)$ for $\S$. If
an automatic structure $\S$ is isomorphic to a structure $\S'$, then
$\S$ is called an {\em automatic copy} of $\S'$ and $\S'$ is {\em
  automatically presentable}. In this paper we sometimes abuse the
terminology referring to $\S'$ as simply automatic and calling an
automatic presentation of $\S$ also automatic presentation
of~$\S'$. We also simplify our statements by saying ``given/compute an
automatic structure $\S$'' for ``given/compute an automatic
presentation $\P$ of a structure $\S(\P)$''. The structures $(\N; \leq,
+)$ and $(\Q; \leq)$ are both automatic structures. On the other hand,
$(\N; \times)$ and $(\Q; +)$ have no automatic copies (see
\cite{KhMi08,Rub08} and \cite{Tsankov}).

Consider $\FO+\exists^\infty+\exists^{n,m}$, the first-order
logic extended by the quantifiers $\exists^{\infty}$ (there exist infinitely
many) and $\exists^{n,m} x$ (there exist finitely many and the exact number
is congruent $n$ modulo $m$, where $m,n\in \N$). The
following theorem from \cite{BluGrae00,Hod82,KhoN95,Rub04}
lays out the main motivation for investigating automatic structures.

\begin{theorem}\label{thm:extFOaut} 
  From an automatic presentation $\P$ and a formula $\varphi(\bar
  x)\in \FO+\exists ^{\infty} +\exists^{n,m}$ in the signature of
  $\S(\P)$, one can compute an automaton whose language consists of
  those tuples $\bar{a}$ from $\S(\P)$ that make $\varphi$ true. In
  particular, the $\FO+\exists^\infty+\exists^{n,m}$ theory of any
  automatic structure $\S$ is (uniformly) decidable. 
\end{theorem}
Let $\K$ be a class of automatic structures closed under
isomorphism. The {\em isomorphism problem} for~$\K$ is the set of
pairs $(\P_1,\P_2)$ of automatic presentations with
$\S(\P_1)\cong\S(\P_2)\in\K$. The isomorphism problem for the class 
of all automatic structures is complete for $\Sigma^1_1$ --- the first level
of the analytical hierarchy \cite{KhoNRS07} (this holds already for
automatic successor trees). However, if one restricts to special subclasses of
automatic structures, this complexity bound can be reduced. For
example, for the class of automatic ordinals and also the class of
automatic Boolean algebras, the isomorphism problem is decidable.
Another interesting result is that the isomorphism problem for locally
finite automatic graphs is $\Pi^0_3$-complete~\cite{Rub04}.
All these classes of
automatic structures have the nice property that one can decide
whether a given automatic presentation describes a structure from
this class. Theorem~\ref{thm:extFOaut} implies that this property
also holds for the classes of equivalence relations, trees of
height at most $k$, and linear orders, i.e., the classes considered
in this paper.

\section{Automatic Equivalence Structures}\label{sec:equiv}

An equivalence structure is of the form $\E=(D; E)$ where $E$ is an
equivalence relation on $D$.
In this section, we prove that the isomorphism problem for automatic
equivalence structures is $\Pi^0_1$-complete. This result can
be also deduced from our result for automatic trees (Section~\ref{sec:tree}).
But the case of equivalence structures is a good starting point for
introducing our techniques.

Let $\E$ be an automatic equivalence structure. Define the function
$h_\E:\N\cup\{\aleph_0\} \rightarrow \N\cup \{\aleph_0\}$ such that
for all $n\in \N\cup \{\aleph_0\}$, $h_\E(n)$ equals the number of
equivalence classes (possibly infinite) in $\E$ of size $n$. Note
that for given $n\in \N\cup \{\aleph_0\}$, the value $h_\E(n)$ can be
computed effectively: one can define in $\FO+\exists^\infty$ the set
of all $\leq_{\text{llex}}$-least elements\footnote{$\leq_{\text{llex}}$ denotes the
  length-lexicographical order on words.} that belong to an
equivalence class of size $n$.

Given two automatic equivalence structures $\E_1=(D_1; E_1)$ and
$\E_2=(D_2;E_2)$, deciding if $\E_1\cong \E_2$ amounts to checking if
$h_{\E_1}=h_{\E_2}$. Therefore, the isomorphism problem for automatic
equivalence structures is in $\Pi^0_1$.

For the $\Pi^0_1$ lower bound, we use a reduction from
Hilbert's $10^{th}$ problem: Given a polynomial
$p(x_1,\ldots,x_k)\in \Z[x_1,\ldots,x_k]$, decide whether the equation
$p(x_1,\ldots,x_k)=0$ has a solution in $\N_+$ (for technical reasons,
it is useful to exclude $0$ in solutions).
This problem is well-known to be
undecidable, see e.g. \cite{Mat93}. In fact, Matiyasevich constructed
from a given (index of a) recursively enumerable set
$X\subseteq \N_+$ a polynomial $p(x_1,\ldots,x_k)\in \Z[x_1,\ldots,x_k]$
such that for all $n \in \N_+$:
$n\in X$ if and only if $\exists y_2,\ldots,y_k \in \N_+ : p(n,y_2,\ldots,y_k)=0$.
Hence, the following set is $\Pi^0_1$-complete:
$$
\{(p_1(\overline x),p_2(\overline x))\in\N[x_1,\dots,x_k]^2 \mid \forall\overline
c\in\N_+^k : p_1(\overline c)\neq p_2(\overline c)\} .
$$
For a symbol $a$, let $\Sigma_k^a$ denote the alphabet
$$
\Sigma_k^a = \{a,\diamond\}^k \setminus \{(\diamond,\ldots,\diamond)\}
$$
and let $\sigma_i$ denote the $i^{th}$ component of $\sigma \in \Sigma_k^a$.
For $\overline e = (e_1,\ldots,e_k) \in\N_+^k$, write
$a^{\overline e}$ for the word
\[
  a^{e_1} \otimes a^{e_2} \otimes \cdots \otimes a^{e_k}\ .
\]
For a language $L$, we write $\otimes_k(L)$ for the language
$$
\{ u_1 \otimes u_2 \otimes \cdots \otimes u_k \mid u_1,\ldots,u_k \in L\}.
$$

\begin{lemma}\label{lm:equiv-runs}
  There exists an algorithm that, given a non-zero polynomial
  $p(\overline{x})\in\N[\overline{x}]$ in $k$ variables, constructs a
  non-deterministic automaton $\A[p(\overline{x})]$ on the alphabet
  $\Sigma_k^a$ with $L(\A[p(\overline x)])=\otimes_k(a^+)$ such that
  for all $\overline c \in \N^k_+$: $\A[p(\overline{x})]$ has exactly
  $p(\overline c)$ accepting runs on input $a^{\overline c}$.
\end{lemma}

\begin{proof}
  The automaton $\A[p(\overline x)]$ is build by induction on the
  construction of the polynomial $p$, the base case is provided by the
  polynomials $1$ and $x_i$.

  Let $\A[1]$ be a deterministic automaton accepting $\otimes_k(a^+)$.
  Next, suppose $p(x_1,\ldots,x_k)=x_i$ for some $i\in \{1,\ldots,k\}$. Let
  $S=\{q_1,q_2\}$, $I=\{q_1\}$ and $F=\{q_2\}$. Define $\Delta$ as
  $$
  \Delta = \{ (q_1,\sigma,q_j) \mid j \in \{1,2\}, \sigma \in \Sigma_k^a,
  \sigma_i = a \} \cup \{ (q_2,\sigma,q_2) \mid  \sigma\in \Sigma_k^a\}.
  $$
  When the automaton $\A[p(\overline x)]=(S,I,\Delta,F)$ runs on an
  input word $a^{\overline c}$, it has exactly $c_i$ many times the
  chance to move from state $q_1$ to the final state $q_2$.  Therefore
  there are exactly $c_i=p(\overline c)$ many accepting runs on
  $a^{\overline c}$.

  Let $p_1(\overline x)$ and $p_2(\overline x)$ be polynomials in
  $\N[\overline{x}]$. Assume as inductive hypothesis that there are
  two automata $\A[p_1(\overline{x})]$ and
  $\A[p_2(\overline{x})]$ such that for $i\in
  \{1,2\}$ the number of accepting runs of $\A[p_i(\overline{x})]$ on
  $a^{\overline c}$ equals $p_i(\overline c)$.

  For $p(\overline x)=p_1(\overline x)+p_2(\overline x)$, set
  $\A[p(\overline{x})] = \A[p_1(\overline{x})] \uplus
  \A[p_2(\overline{x})]$.  Then, the number of accepting runs of
  $\A[p(\overline{x})]$ on $a^{\overline c}$ is $p_1(\overline
  c)+p_2(\overline c)$.

  For $p(\overline x)=p_1(\overline x)\cdot p_2(\overline x)$, let
  $\A[p(\overline{x})] = \A[p_1(\overline{x})] \times
  \A[p_2(\overline{x})]$.  Then, the number of accepting runs of
  $\A[p(\overline{x})]$ on $a^{\overline c}$ is $p_1(\overline c)
  \cdot p_2(\overline c)$.  \qed
\end{proof}

\noindent
Let $\A=(S,I,\Delta,F)$ be a non-deterministic finite automaton with
alphabet $\Sigma$. We define an automaton $\Run_\A=(S,I,\Delta',F)$
with alphabet $\Delta$ and
$$
\Delta' =  \{ (p, (p,a,q), q) \mid (p,a,q) \in \Delta \}.
$$
Let $\pi : \Delta^* \to \Sigma^*$ be the projection morphism with
$\pi(p,a,q) = a$.
The following lemma is immediate from the definition.

\begin{lemma}
  For $u \in \Delta^+$ we have: $u \in L(\Run_\A)$ if and only if $u$
  forms an accepting run of $\A$ on $\pi(u)$ (which in particular
  implies $\pi(u) \in L(\A)$).
\end{lemma}
This lemma implies that for all words $w\in \Sigma^+$, $|\pi^{-1}(w)
\cap L(\Run_\A)|$ equals the number of accepting runs of $\A$ on $w$.
Note that this does not hold for $w = \varepsilon$.

Consider a non-zero polynomial $p(\overline x)\in
\N[x_1,\ldots,x_k]$. Let the automaton $\A=\A[p(\overline{x})]$
satisfy the properties guaranteed by Lemma~\ref{lm:equiv-runs} and let
$\Run_{\A}$ be as defined above. Define an automatic equivalence
structure $\E(p)$ whose domain is $L(\Run_\A)
\setminus \{\varepsilon\}$.  Moreover,
two words $u,v \in L(\Run_\A)\setminus\{\varepsilon\}$ are equivalent
if and only if $\pi(u) = \pi(v)$. By definition and
Lemma~\ref{lm:equiv-runs}, a natural number $y \in \N_+$ belongs to
$\Img_+(p)$ if and only if there exists a word $u\in L(\A)$ with
precisely $y$ accepting runs, if and only if $\E(p)$ contains an
equivalence class of size $y$.

It is well known that the function $C: \N \times \N \to \N$ with
\begin{equation} \label{function-C}
    C(x,y) = (x+y)^2+3x+y
\end{equation}
is injective ($C(x,y)/2$ defines a pairing function, see e.g. \cite{Hon06}).  In the
following, let $\E_\Good$ denote the countably infinite equivalence
structure with
\[
h_{\E_\Good}(n)=
  \begin{cases}
    \infty & \text{if }n\in\{C(y,z)\mid y,z\in\N_+, y\neq z\}\\
    0      & \text{otherwise.}
  \end{cases}
\]

\begin{proposition}\label{P-equiv}
  The set of automatic presentations $\P$ with $\S(\P)\cong\E_\Good$ is
  hard for $\Pi^0_1$.
\end{proposition}

\begin{proof}
  For non-zero polynomials $p_1(\overline x),p_2(\overline x)\in
  \N[x_1,\ldots,x_k]$, define the following three (non-zero)
  polynomials from $\N[x_1,\ldots,x_k]$ (with $k\ge2$):
  \begin{align*}
    S_1(\overline{x})& = C(p_1(\overline{x}), p_2(\overline{x})), &
    S_2(\overline x) &= C(x_1+x_2,x_1), &
    S_3(\overline x) &= C(x_1, x_1+x_2) .
  \end{align*}
  Let $\E(S_1)$, $\E(S_2)$, and $\E(S_3)$ be the automatic equivalence
  structures corresponding to these polynomials according to the above
  definition. Finally, let $\E$ be the disjoint union of
  $\aleph_0$ many copies of these three equivalence structures.

  If $p_1(\overline c)=p_2(\overline c)$ for some $\overline c\in
  \N_+^k$, then there is $y\in \N_+$ such that $C(y,y) \in \Img_+(S_1)$.
  Therefore in $\E$ there is an equivalence class of size $C(y,y)$ and
  no such equivalence class exists in $\E_\Good$. Hence $\E
  \ncong \E_\Good$.

  Conversely, suppose that $p_1(\overline c)\neq p_2(\overline c)$ for
  all $\overline c\in \N_+^k$.  For all $y,z\in \N_+$, $\E$ contains
  an equivalence class of size $C(y,z)$ if and only if $C(y,z)$
  belongs to $\Img_+(S_1) \cup \Img_+(S_2) \cup \Img_+(S_3)$, if and
  only if $y\neq z$, if and only if $\E_\Good$ contains an equivalence
  class of size $C(y,z)$. Therefore, for any $s\in \N_+$, $\E$ contains
  an equivalence class of size $s$ if and only if $\E_\Good$ contains
  an equivalence class of size $s$. Hence $\E\cong \E_\Good$.

  In summary, we have reduced the $\Pi^0_1$-hard problem 
  \[
     \{(p_1(\overline x),p_2(\overline x))\in\N[x_1,\dots,x_k]^2 
      \mid k\ge2,\forall\overline c\in\N_+^k:p_1(\overline c)\neq p_2(\overline c)\}
  \]
  to the set of automatic presentations of $\E_\Good$.
  Hence the proposition is proved.
\qed
\end{proof} 

\begin{theorem}\label{thm:equiv}
  The isomorphism problem for automatic equivalence structures is
  $\Pi^0_1$-complete.
\end{theorem}

\begin{proof}
  At the beginning of this section, we already argued that the
  isomorphism problem is in~$\Pi^0_1$; hardness follows immediately
  from Proposition~\ref{P-equiv}, since $\E_\Good$ is necessarily automatic.
\qed
\end{proof}

\section{Automatic Trees}\label{sec:tree}

A {\em tree} is a structure $T=(V;\leq)$, where $\leq$ is a
partial order with a least element, called the {\em root}, and such
that for every $x \in V$, the order $\leq$ restricted to the set
$\{y\mid y\leq x\}$ of ancestors of $x$ is a finite linear order.
The {\em level} of a node $x\in V$ is
$|\{y \mid y < x \}| \in \mathbb{N}$. The {\em height} of
$T$ is the supremum of the levels of all nodes in $V$; it may be
infinite, but this paper deals with trees of finite height only.
One may also view a tree as a directed graph $(V,E)$, where there is
an edge $(u,v)\in E$ if and only if $u$ is the largest element in $\{
x \mid x < v \}$.  The edge relation $E$ is FO-definable in
$(V;\leq)$. In this paper, we assume the partial order definition for
trees, but will quite often refer to them as graphs for
convenience.  We use $\T_n$ to denote the class of automatic trees with height at most
$n$. Let $n$ be fixed. Then the tree order $\leq$ is FO-definable in
$T$ and this holds even uniformly for all trees from $\T_n$. Moreover,
it is decidable whether a given automatic graph belongs to $\T_n$
(since the class of trees of height $n$ can be axiomatized in
first-order logic). 

As a corollary to Proposition~\ref{P-equiv}, we get immediately that the
isomorphism problem for automatic trees of height at most $2$ is
undecidable:

\begin{corollary}\label{crl:tree2}
  There exists an automatic tree $T_\Good$ of height $2$ such that the set of
  automatic presentations~$\P$ with $\S(\P)\cong T_\Good$ is
  $\Pi^0_1$-hard. Hence, the isomorphism problem for the class $\T_2$
  of automatic trees of height at most $2$ is $\Pi^0_1$-hard.
\end{corollary}

\begin{proof}
  Let $\E=(V;\equiv)$ be an automatic equivalence structure.
  Now build the tree $T(\E)$ as follows:
  \begin{itemize}
  \item the set of nodes is $V\cup \{r\} \cup \{au \mid u \in V, u\text{
      is $\leq_{\text{llex}}$-minimal in }[u]_\equiv\}$ where $r$ and
    $a$ are two new letters
  \item $r$ is the root, its children are the words starting with $a$,
    and the children of $au$ are the words from $[u]_\equiv$.
  \end{itemize}
  Then it is clear that $T(\E)$ is a tree of height at most $2$ and
  that an automatic presentation for $T(\E)$ can be computed from one
  for~$\E$. Furthermore, $\E\cong \E_\Good$ if and only if $T(\E)\cong
  T(\E_\Good)$. Hence, indeed, the statement follows from
  Proposition~\ref{P-equiv}.\qed
\end{proof}
The hardness statement of Theorem~\ref{thm:tree} below is a
generalization of this corollary to all the classes $\T_n$ for
$n\ge2$. But first, we prove an upper bound for the isomorphism
problem for $\T_n$:

\begin{proposition}
\label{prop:tree_membership}
  The isomorphism problem for the class $\T_n$ of automatic trees of
  height at most $n$ is
  \begin{itemize}
  \item decidable for $n=1$ and
  \item in $\Pi^0_{2n-3}$ for all $n \geq 2$.
  \end{itemize}
\end{proposition}

\begin{proof}
  We first show that $T_1\cong T_2$ is decidable for automatic trees
  $T_1, T_2 \in \T_1$ of height at most~$1$: It suffices to compute
  the cardinality of $T_i$ ($i \in \{1,2\}$) which is possible since
  the universes of $T_1$ and $T_2$ are regular languages.

  Now let $n\ge2$ and consider $T_1,T_2\in \T_n$.
  Let $T_i = (V_i,E_i)$, w.l.o.g. $V_1 \cap V_2 = \emptyset$, and $V =
  V_1 \cup V_2$, $E = E_1 \cup E_2$.  For any node $u$ in $V$, let
  $T(u)$ denote the subtree (of either $T_1$ or $T_2$) rooted at $u$
  and let $E(u)$ be the set of children of $u$.  For
  $k=n-2,n-3,\dots,0$, we will define inductively a
  $\Pi^0_{2n-2k-3}$-predicate $\text{iso}_k(u_1,u_2)$ for $u_1,u_2\in
  V$. This predicate expresses that $T(u_1) \cong T(u_2)$ provided
  $u_1$ and $u_2$ belong to level at least $k$.  The result will follow since
  $T_1\cong T_2$ if and only if $\text{iso}_0(r_1,r_2)$ holds, where $r_\sigma$ is
  the root of $T_\sigma$.

  For $k=n-2$, the trees $T(u_1)$ and $T(u_2)$ have height at most $2$
  and we can define $\text{iso}_{n-2}(u_1,u_2)$ as follows:
  \[
    \forall \kappa \in \N \cup \{ \aleph_0 \} \
    \forall \ell \geq 1 \
    \left(
    \begin{array}{l}
      \phantom{\iff} \displaystyle \exists x_1,\ldots,x_\ell \in E(u_1) :
      \bigwedge_{1\leq i<j\leq \ell} x_i \neq x_j \wedge
      \bigwedge_{i=1}^\ell |E(x_i)|= \kappa \\
       \iff  \displaystyle\exists y_1,\ldots,y_\ell \in E(u_2) :
      \bigwedge_{1\leq i<j\leq \ell}y_i \neq y_j \wedge
      \bigwedge_{i=1}^\ell |E(y_i)|=\kappa
    \end{array}
    \right)
  \]
  In other words: for every $\kappa \in \N \cup \{ \aleph_0 \}$, $u_1$
  and $u_2$ have the same number of children with exactly $\kappa$
  children. Since $\FO+\exists^\infty$ is uniformly decidable for
  automatic structures, this is indeed a $\Pi^0_1$-sentence (note that
  $2n-2k-3 = 1$ for $k=n-2$). For $0 \leq k < n-2$, we define
  $\text{iso}_k(u_1,u_2)$ inductively as follows:
  \[
    \forall v \in E(u_1) \cup E(u_2) \ \forall \ell \geq 1
    \left(
    \begin{array}{l}
      \displaystyle \phantom{\iff} \exists x_1,\ldots,x_\ell \in E(u_1) :
      \bigwedge_{1\leq i<j\leq \ell} x_i \neq x_j \wedge
      \bigwedge_{i=1}^\ell \text{iso}_{k+1}(v,x_i)  \\
       \iff  \displaystyle\exists y_1,\ldots,y_\ell \in E(u_2) :
      \bigwedge_{1\leq i<j\leq \ell}y_i \neq y_j \wedge
      \bigwedge_{i=1}^\ell \text{iso}_{k+1}(v,y_i)
    \end{array}
    \right)
  \]
  By quantifying over all $v \in E(u_1) \cup E(u_2)$, we quantify over
  all isomorphism types of trees that occur as a subtree rooted at a
  child of $u_1$ or $u_2$. For each of these isomorphism types $\tau$,
  we express that $u_1$ and $u_2$ have the same number of children~$x$
  with $T(x)$ of type $\tau$.  Since by induction,
  $\text{iso}_{k+1}(v,x_i)$ and $\text{iso}_{k+1}(v,y_i)$ are
  $\Pi^0_{2n-2k-5}$-statements, $\text{iso}_k(u_1,u_2)$ is a
  $\Pi^0_{2n-2k-3}$-statement.  \qed
\end{proof}
The rest of this section is devoted to proving that the isomorphism
problem for the class $\T_n$ of automatic trees of height at most
$n\ge2$ is also $\Pi^0_{2n-3}$-hard (and therefore complete). So let
$P_n(x_0)$ be a $\Pi^0_{2n-3}$-predicate.
In the following lemma and its proof, all quantifiers with 
unspecified range run over $\mathbb{N}_+$.
\begin{lemma}\label{lem:normalform}
  For $2\leq i\leq n$, there are $\Pi^0_{2i-3}$-predicates
  $P_i(x_0,x_1,y_1,x_2,y_2,\ldots,x_{n-i},y_{n-i})$ such that
  \begin{enumerate}
  \item[(i)] $P_{i+1}(\overline x)$ is logically equivalent to
    $\forall x_{n-i}\exists y_{n-i}: P_i(\overline x,x_{n-i},y_{n-i})$
    for $2 \leq i < n$ and
  \item[(ii)] $\forall y_{n-i} : \neg P_i(\overline{x},x_{n-i},y_{n-i})$
    implies $\forall x_{n-i}'\geq x_{n-i}\ \forall y_{n-i} : \neg
    P_i(\overline{x},x_{n-i}',y_{n-i})$,
  \end{enumerate}
  where $\overline x=(x_0,x_1,y_1,\dots,x_{n-i-1},y_{n-i-1})$.
\end{lemma}

\begin{proof}
  The predicates $P_i$ are constructed by induction, starting with
  $i=n-1$ down to $i=2$ where the construction of $P_i$ does not
  assume that (i) or (ii) hold true for $P_{i+1}$.

  So let $2\le i<n$ such that $P_{i+1}(\overline x)$ is a
  $\Pi^0_{2(i+1)-3}$-predicate. Then there exists a
  $\Pi^0_{2i-3}$-predicate $P(\overline x,x_{n-i},y_{n-i})$ such that
  $P_{i+1}(\overline x)$ is logically equivalent to
  \[
     \forall x_{n-i}\exists y_{n-i}: P(\overline x,x_{n-i},y_{n-i})\ .
  \]
  But this is logically equivalent to
  \begin{equation}\label{eqt:tree_k}
     \forall x_{n-i}\ \forall x_{n-i}'\leq x_{n-i} \ \exists y_{n-i} :
         P(\overline{x},x_{n-i}',y_{n-i})\ .
  \end{equation}
  Let $\varphi(\overline{x},x_{n-i})$ be
  \[
     \forall x_{n-i}'\leq x_{n-i} \ \exists y_{n-i} :
         P(\overline{x},x_{n-i}',y_{n-i})\ .
  \]
  Then for any $x_{n-i}\in \N$,
  \begin{equation}\label{eqt:tree_neg}
    \neg \varphi(\overline{x},x_{n-i}) \ \Longrightarrow \
       \forall x \geq x_{n-i} : \neg \varphi(\overline{x},x)\ .
  \end{equation}
  Since $\forall x_{n-i}'\leq x_{n-i}$ is a bounded quantifier, the
  formula $\varphi(\overline{x},x_{n-i})$ belongs to $\Sigma^0_{2i-2}$
  (see for example \cite[p. 61]{Soa87}).  Thus there is a
  $\Pi^0_{2i-3}$-predicate $P_i(\overline{x},x_{n-i},y_{n-i})$ such that
  \begin{equation} \label{equiv-varphi-P_i}
    \varphi(\overline{x},x_{n-i})  \ \Longleftrightarrow \ \exists y_{n-i} : P_i(\overline{x},x_{n-i},y_{n-i})\ .
  \end{equation}
  Therefore (\ref{eqt:tree_k}) (and therefore $P_{i+1}(\overline{x})$)
  is logically equivalent to
  $\forall x_{n-i} \ \exists y_{n-i} : P_i(\overline{x},x_{n-i},y_{n-i})$.
  Moreover,
  \begin{eqnarray*}
  \forall y_{n-i} : \neg P_i(\overline{x},x_{n-i},y_{n-i}) 
  \ & \stackrel{\text{(\ref{equiv-varphi-P_i})}}{\Longleftrightarrow} & \
  \neg \varphi(\overline{x},x_{n-i}) \\
  &  \stackrel{\text{(\ref{eqt:tree_neg})}}{\Longrightarrow} & \
  \forall x \geq x_{n-i} : \neg \varphi(\overline{x},x) \\
  & \stackrel{\text{(\ref{equiv-varphi-P_i})}}{\Longleftrightarrow} & \
     \forall x \geq x_{n-i}\ \forall y_{n-i} : \neg
    P_i(\overline{x},x,y_{n-i})
 \end{eqnarray*} 
 This shows (ii).  \qed
\end{proof}
Let us fix the predicates $P_i$ for the rest of Section~\ref{sec:tree}.
By induction on $2\le i\le n$, we will construct the following trees:
\begin{itemize}
\item test trees $T^i_{\overline c} \in \T_i$ for $\overline
  c\in\N_+^{1+2(n-i)}$ (which depend on $P_i$) and
\item trees $U^i_\kappa \in \T_i$ for $\kappa \in \N_+ \cup
  \{\omega\}$ (we assume the standard order on $\N_+ \cup \{\omega\}$).
\end{itemize}
The idea is that $T^i_{\overline c} \cong U^i_\kappa$ if and only if
$\kappa=1+\inf(\{x_{n-i}\mid \forall y_{n-i}\in\N_+:\neg P_i(\overline
c,x_{n-i},y_{n-i})\}\cup\{\omega\})$. We will not prove this equivalence, but
the following simpler consequences for any $\overline
c\in\N_+^{1+2(n-i)}$:
\begin{description}
\item[(P1)] $P_i(\overline c)$ holds if and only if
  $T^i_{\overline c} \cong U^i_\omega$.
\item[(P2)] $P_i(\overline c)$ does not hold if and only if
  $T^i_{\overline c} \cong U^i_m$ for some $m \in \N_+$.
\end{description}
The first property is certainly sufficient for proving
$\Pi^0_{2n-3}$-hardness (with $i=n$), the second property and
therefore the trees $U^i_m$ for $m<\omega$ are used in the inductive
step. We also need the following property for the construction.
\begin{description}
\item[(P3)] No leaf of any of the trees $T^i_{\overline{c}}$ or
  $U^i_{\kappa}$ is a child of the root.
\end{description}
In the following section, we will describe the trees
$T^i_{\overline{c}}$ and $U^i_\kappa$ of height at
most $i$ and prove~(P1) and (P2). 
Condition (P3) will be obvious from the construction.
The subsequent section is then devoted to
prove the effective automaticity of these trees.

\subsection{Construction of trees}
\label{Construction of trees}

We start with a few definitions: 
A forest is a disjoint union of trees. 
Let $H_1$ and $H_2$ be two forests.
The forest $H_1^\omega$ is the disjoint union of countably many copies
of $H_1$. Formally, if $H_1=(V,E)$, then $H_1^\omega=(V\times\N,E')$
with $((v,i),(w,j))\in E'$ if and only if $(v,w)\in E$ and $i=j$. We write
$H_1\sim H_2$ for $H_1^\omega\cong H_2^\omega$.
Then $H_1 \sim H_2$ if they are formed, up to isomorphism, by
the same set of trees (i.e., any tree is isomorphic to some connected
component of $H_1$ if and only if it is isomorphic to some connected component of
$H_2$). 
If $H$ is a forest and $r$ does not belong to the domain of
$H$, then we denote with $r \circ H$ the tree that
results from adding $r$ to $H$ as new least element.

\subsubsection{Induction base: construction of $T^2_{\overline{c}}$ and
  $U^2_\kappa$}\label{sss:base}

For notational simplicity, we write $k$ for $1+2(n-2)$. 
Hence, $P_2$ is a $k$-ary predicate. By
Matiyasevich's theorem, we find two non-zero polynomials
$p_1(x_1,\ldots,x_\ell)$, $p_2(x_1,\ldots,x_\ell)\in
\N[\overline{x}]$, $\ell> k$, such that for any $\overline{c} \in
\N_+^k$:
\[
 P_2(\overline{c}) \text{ holds  } \ \iff \ \forall \overline{x} \in \N_+^{\ell-k} :
  p_1(\overline{c},\overline{x})\neq p_2(\overline{c},\overline{x})\ .
\]
For two numbers $m,n\in\N_+$, let $T[m,n]$ denote the tree of height
$1$ with exactly $C(m,n)$ leaves, where $C$ is the injective polynomial function 
from (\ref{function-C}).  Then define the following forests:
\begin{align*}
  H^2 &= \biguplus\{T[m,n]\mid m,n\in\N_+,m\neq n\}\\
  H^2_{\overline c} &=  H^2 \uplus\biguplus
\{T[p_1(\overline c,\overline x)+x_{\ell+1},
    p_2(\overline c,\overline x)+x_{\ell+1}]\mid \overline x\in\N_+^{\ell-k},x_{\ell+1}\in\N_+\}\\
  J^2_\kappa &= H^2 \uplus\biguplus\{T[x,x]\mid x\in\N_+, x > \kappa\}  \quad
  \text{ for }\kappa\in\N_+   \cup \{\omega\}
\end{align*}
Note that $J^2_\omega = H^2$. Moreover, the forests $J^2_\kappa$
($\kappa\in\N_+   \cup \{\omega\}$) are pairwise non-isomorphic, since $C$
is injective.

The trees $T^2_{\overline{c}}$ and $U^2_\kappa$, resp.,
are obtained from $H^2_{\overline c}$ and $J^2_\kappa$, resp., by
taking countably many copies and adding a root:
\begin{equation}\label{def-tree-from-forest-2}
  T^2_{\overline c} = r\circ (H^2_{\overline c})^\omega \qquad
  U^2_\kappa        = r\circ (J^2_\kappa)^\omega,
\end{equation}
see Figure~\ref{fig:tree1}.

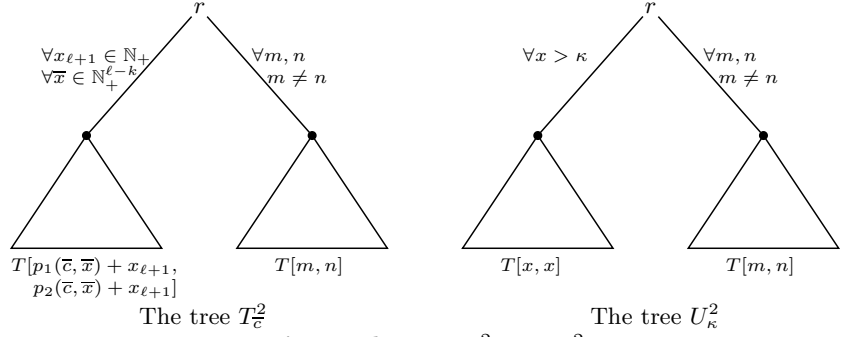
\begin{figure}[t]
\begin{center}
\setlength{\unitlength}{1mm}
\begin{picture}(20,35)(0,5)
\gasset{Nadjustdist=0.8,Nadjust=wh,Nframe=n,Nfill=n,AHnb=0,ELdist=0.5,linewidth=.2}
   \put(-28,-2){The tree $T^2_{\overline{c}}$}
   \node(root)(-20,40){$r$}
 \gasset{Nw=1,Nh=1,Nframe=y,Nfill=y,Nadjust=n}
   \node(a)(-35,23){}
    \node(a'')(-5,23){}
   \drawedge[ELside=r](root,a){}
   \put(-41,30){\text{\scriptsize $\forall \overline{x}\in \N_+^{\ell-k}$}}
   \put(-41,33){\text{\scriptsize $\forall x_{\ell+1}\in \N_+$}}
    \drawedge[ELside=l,ELpos=60,ELdist=0](root,a''){}
    \put(-13,33){\text{\scriptsize $\forall m,n$}}
    \put(-11,30){\text{\scriptsize $m\neq n$}}

   \drawpolygon[Nframe=y,Nfill=n](-35,23)(-45,8)(-25,8)
   \put(-45,5){\text{\scriptsize $T[p_1(\overline{c},\overline{x})+x_{\ell+1},$}}
   \put(-42.1,2){\text{\scriptsize $p_2(\overline{c},\overline{x})+x_{\ell+1}]$}}
   \drawpolygon[Nframe=y,Nfill=n](-5,23)(-15,8)(5,8)
   \put(-10,5){\text{\scriptsize $T[m,n]$}}


\gasset{Nadjustdist=0.8,Nadjust=wh,Nframe=n,Nfill=n,AHnb=0,ELdist=0.5,linewidth=.2}
   \put(32,-2){The tree $U^2_{\kappa}$}
   \node(root)(40,40){$r$}
\gasset{Nw=1,Nh=1,Nframe=y,Nfill=y,Nadjust=n}
   \node(a)(25,23){}
    \node(a'')(55,23){}
   \drawedge[ELside=r](root,a){}
   \put(23,33){\text{\scriptsize $\forall x>\kappa$}}
    \drawedge[ELside=l,ELpos=60,ELdist=0](root,a''){}
    \put(47,33){\text{\scriptsize $\forall m,n$}}
    \put(49,30){\text{\scriptsize $m\neq n$}}
   \drawpolygon[Nframe=y,Nfill=n](25,23)(15,8)(35,8)
   \put(20,5){\text{\scriptsize $T[x,x]$}}
   \drawpolygon[Nframe=y,Nfill=n](55,23)(45,8)(65,8)
   \put(50,5){\text{\scriptsize $T[m,n]$}}

\end{picture}
\end{center}
\caption{\label{fig:tree1} The tree $T^2_{\overline{c}}$ and $U^2_{\kappa}$}
\end{figure}

The following lemma (stating (P1) for the $\Pi^0_1$-predicate~$P_2$ ,
i.e., for $i=2$) can be proved in a similar way as
Theorem~\ref{thm:equiv}.

\begin{lemma}\label{lem:tree_2_good}
  For any $\overline c\in \N_+^k$, we have:
  \[
  P_2(\overline c) \text{ holds } \iff
      H^2_{\overline{c}}\sim J^2_\omega \iff
      T^2_{\overline{c}}\cong U^2_\omega\ .
  \]
\end{lemma}

\begin{proof}
  By (\ref{def-tree-from-forest-2}),
  it suffices to show the first equivalence. So first
  assume $P_2(\overline c)$ holds. We have to prove that the forests
  $H^2_{\overline c}$ and $J^2_\omega = H^2$ contain the same trees (up to
  isomorphism). Clearly, every tree from $H^2$ is contained in $H^2_{\overline c}$.
  For the other direction, let $\overline x\in\N_+^{\ell-k}$ and
  $x_{\ell+1}\in\N_+$. Then the tree $T[p_1(\overline c,\overline
  x)+x_{\ell+1},p_2(\overline c,\overline x)+x_{\ell+1}]$ occurs in
  $H^2_{\overline c}$. Since $P_2(\overline c)$ holds, we have
  $p_1(\overline c,\overline x)\neq p_2(\overline c,\overline x)$ and
  therefore $p_1(\overline c,\overline x)+x_{\ell+1}\neq p_2(\overline
  c,\overline x)+x_{\ell+1}$. Hence this tree also occurs in
  $H^2$.

  Conversely suppose $H^2_{\overline c}\sim H^2$ and let
  $\overline x\in\N_+^{\ell-k}$. Then the tree $T[p_1(\overline
  c,\overline x)+1,p_2(\overline c,\overline x)+1]$ occurs in
  $H^2_{\overline c}$ and therefore in $H^2$. Hence
  $p_1(\overline c,\overline x)\neq p_2(\overline c,\overline
  x)$. Since $\overline x$ was chosen arbitrarily, this
  implies~$P_2(\overline c)$.\qed
\end{proof}
Now consider the forest $H^2_{\overline c}$ once more.
If it contains a tree of the form $T[m,m]$ for some $m$ 
(necessarily $m \geq 2$), then it contains all 
trees $T[x,x]$ for $x\geq m$. Hence, $H^2_{\overline{c}} \sim 
J^2_\kappa$ for some $\kappa \in
\mathbb{N}_+ \cup \{\omega\}$, which implies
$T^2_{\overline{c}} \cong  U^2_\kappa$ for some $\kappa \in
\mathbb{N}_+ \cup \{\omega\}$.
Thus, with Lemma~\ref{lem:tree_2_good} we get:
$$
  P_2(\overline c)\text{ does not hold} \ \iff \
  T^2_{\overline{c}} \not\cong U^2_\omega \ \iff \
  \exists m \in \mathbb{N}_+ :  T^2_{\overline{c}} \cong U^2_m
$$
Hence we proved the following lemma, which states
(P2) for the $\Pi^0_1$-predicate~$P_2$, i.e., for $i=2$.

\begin{lemma}\label{lem:tree_2_bad}
  For any $\overline c\in \N_+^k$, we have:
  $$
  P_2(\overline c) \text{ does not hold } \  \iff  \
     \exists m \in \N_+ : T^2_{\overline{c}} \cong U^2_m \ .
  $$
\end{lemma}
This finishes the construction of the trees $T^2_{\overline c}$ and
$U^2_\kappa$ for $\kappa\in\N_+\cup\{\omega\}$, and the verification
of properties~(P1) and (P2).
Clearly, also (P3) holds for $T^2_{\overline c}$ and
$U^2_\kappa$ (all maximal paths have length 2).

\subsubsection{Induction step: construction of $T^{i+1}_{\overline{c}}$
and  $U^{i+1}_\kappa$} \label{sec-induction-trees}

For notational simplicity, we write again $k$ for $1+2(n-i-1)$ such
that $P_{i+1}$ is a $k$-ary predicate and $P_i$ a $(k+2)$-ary one.

We now apply the induction hypothesis. For any
$\overline{c} \in \N_+^k$, $x,y \in \N_+$, $\kappa \in \N_+ \cup \{\omega\}$
let $T^{i}_{\overline{c}xy}$ and
$U^i_\kappa$ be trees of height at most~$i$
such that:
\begin{itemize}
\item $P_i(\overline{c},x,y)$ holds if and only if
  $T^i_{\overline{c}xy} \cong U^i_\omega$.
\item $P_i(\overline{c},x,y)$ does not hold if and only if
  $T^i_{\overline{c}xy} \cong U^i_m$ for some $m \in \N_+$.
\end{itemize}
In a first step, we build the trees $T'_{\overline{c}xy}$ and
$U'_{\kappa,x}$ ($x \in \N_+$) from $T^i_{\overline{c}xy}$
and $U^i_\kappa$, resp., by adding $x$ leaves as
children of the root. This ensures
\begin{eqnarray}
    T'_{\overline{c}xy}\cong T'_{\overline{c}x'y'} &\iff&
      x=x' \wedge T^i_{\overline{c}xy}\cong T^i_{\overline{c}x'y'}\text{ and }
       \label{eqt:tree_T'(P)}\\
    T'_{\overline{c}xy} \cong U'_{\kappa,x'} &\iff&
    x=x' \wedge T^i_{\overline{c}xy}\cong U^i_\kappa\ ,
       \label{eqt:tree_U'}
\end{eqnarray}
since, by property~(P3), no leaf of any of the trees $T^i_{\overline{c}xy}$ or $U^i_\kappa$ is a child of the
root. Next, we collect these trees into forests as follows:
\begin{align*}
  H^{i+1} & = \biguplus\{ U'_{m,x} \mid x,m\in \N_+\}\ ,\\
  H^{i+1}_{\overline c}
    & = H^{i+1}\uplus\biguplus \{ T'_{\overline{c}xy} \mid x,y\in \N_+\}\ ,\text{ and }\\
  J^{i+1}_\kappa & = H^{i+1}\uplus
               \biguplus\{ U'_{\omega,x} \mid 1 \leq x < \kappa \}\text{ for } \kappa \in
               \N_+ \cup \{\omega\} .
\end{align*}
The trees $T^{i+1}_{\overline c}$ and $U^{i+1}_\kappa$, resp., are then obtained from
the forests $H^{i+1}_{\overline c}$ and $J^{i+1}_\kappa$, resp., by
taking countably many copies and adding a root:
\begin{equation}\label{def-tree-from-forest-i+1}
   T^{i+1}_{\overline c} = r\circ (H^{i+1}_{\overline c})^\omega \quad\text{and}\quad
   U^{i+1}_\kappa = r\circ (J^{i+1}_\kappa)^\omega,
\end{equation}
see Figure~\ref{fig:tree2}.

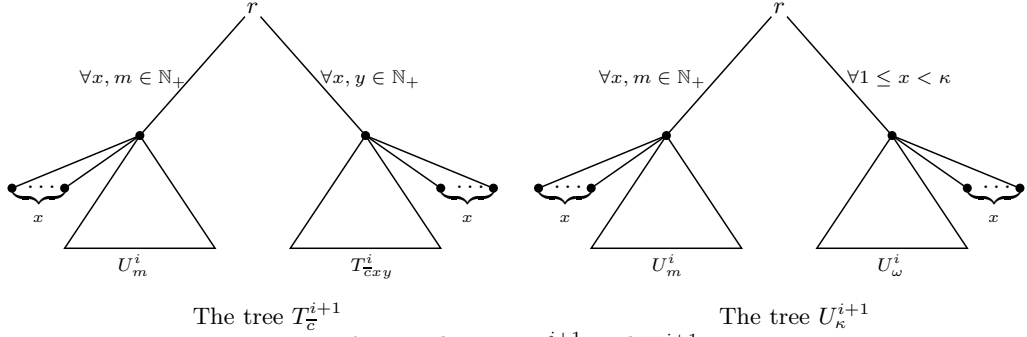
\begin{figure}[t]
\begin{center}
\setlength{\unitlength}{1mm}
\begin{picture}(20,35)(0,5)
\gasset{Nadjustdist=0.8,Nadjust=wh,Nframe=n,Nfill=n,AHnb=0,ELdist=0.5,linewidth=.2}
   \put(-33,-2){The tree $T^{i+1}_{\overline{c}}$}
   \node(root)(-25,40){$r$}
\gasset{Nw=1,Nh=1,Nframe=y,Nfill=y,Nadjust=n}
   \node(a)(-40,23){}
    \node(a'')(-10,23){}
   \drawedge[ELside=r](root,a){}
   \put(-48,30){\text{\scriptsize $\forall x,m\in\N_+$}}
    \drawedge[ELside=l,ELpos=60,ELdist=0](root,a''){}
    \put(-16,30){\text{\scriptsize $\forall x,y\in \N_+$}}
   \drawpolygon[Nframe=y,Nfill=n](-40,23)(-50,8)(-30,8)

   \node(c)(-57,16){}
   \node(c')(-50,16){}
   \drawedge(a,c){}
   \drawedge(a,c'){}
   \put(-55,16){$\ldots$}
   \put(-57,16){$\underbrace{\makebox(7,0){}}_{x}$}
   \put(-43,5){\text{\scriptsize $U^i_m$}}

   \drawpolygon[Nframe=y,Nfill=n](-10,23)(-20,8)(0,8)

   \node(d)(0,16){}
   \node(d')(7,16){}
   \drawedge(a'',d){}
   \drawedge(a'',d'){}
   \put(2,16){$\ldots$}
   \put(0,16){$\underbrace{\makebox(7,0){}}_{x}$}
   \put(-12,5){\text{\scriptsize $T^i_{\overline{c}xy}$}}


\gasset{Nadjustdist=0.8,Nadjust=wh,Nframe=n,Nfill=n,AHnb=0,ELdist=0.5,linewidth=.2}
   \put(37,-2){The tree $U^{i+1}_{\kappa}$}
   \node(root)(45,40){$r$}
\gasset{Nw=1,Nh=1,Nframe=y,Nfill=y,Nadjust=n}
   \node(a)(30,23){}
    \node(a'')(60,23){}
   \drawedge[ELside=r](root,a){}
   \put(21,30){\text{\scriptsize $\forall x,m\in \N_+$}}
    \drawedge[ELside=l,ELpos=60,ELdist=0](root,a''){}
    \put(54,30){\text{\scriptsize $\forall 1\leq x<\kappa$}}

   \drawpolygon[Nframe=y,Nfill=n](30,23)(20,8)(40,8)
   \node(c)(13,16){}
   \node(c')(20,16){}
   \drawedge(a,c){}
   \drawedge(a,c'){}
   \put(15,16){$\ldots$}
   \put(13,16){$\underbrace{\makebox(7,0){}}_{x}$}
   \put(28,5){\text{\scriptsize $U^i_m$}}
   \drawpolygon[Nframe=y,Nfill=n](60,23)(50,8)(70,8)
   \node(d)(70,16){}
   \node(d')(77,16){}
   \drawedge(a'',d){}
   \drawedge(a'',d'){}
   \put(72,16){$\ldots$}
   \put(70,16){$\underbrace{\makebox(7,0){}}_{x}$}
   \put(58,5){\text{\scriptsize $U^i_\omega$}}

\end{picture}
\end{center}
\caption{\label{fig:tree2} The tree $T^{i+1}_{\overline{c}}$ and $U^{i+1}_{\kappa}$}
\end{figure}

Note that the height of any of these trees is one more than the height
of the forests defining them and therefore at most $i+1$. Since none
of the connected components of the forests $H^{i+1}_{\overline c}$ and $J^{i+1}_\kappa$ is a
singleton, none of the trees in (\ref{def-tree-from-forest-i+1}) has a leaf that is a
child of the root and therefore (P3) holds.

\begin{lemma}\label{lem:tree_main1}
  For all $\overline{c} \in \N_+^k$ we have
  \[
    P_{i+1}(\overline{c}) \text{ holds } \iff
     H^{i+1}_{\overline{c}}\sim J^{i+1}_\omega \iff
     T^{i+1}_{\overline{c}}\cong U^{i+1}_\omega\ .
  \]
\end{lemma}

\begin{proof}
  Again, we only have to prove the first equivalence.

  First assume $H^{i+1}_{\overline c}\sim J^{i+1}_\omega$ and let
  $x\ge1$ be arbitrary. We have to exhibit some $y\ge1$ such that
  $P_i(\overline c,x,y)$ holds. Note that $U'_{\omega,x}$ belongs to
  $J^{i+1}_\omega$ and therefore to $H^{i+1}_{\overline c}$. Since
  $U'_{\omega,x}\not\cong U'_{m,x'}$ for any $m,x,x'\in \N_+$, this
  implies the existence of $x',y'\ge1$ with
  $T'_{\overline cx'y'}\cong U'_{\omega,x}$. By (\ref{eqt:tree_U'}), this is
  equivalent with $x=x'$ and $T^i_{\overline cxy'}\cong U^i_\omega$. Now the
  induction hypothesis implies that $P_i(\overline c,x,y')$ holds. Since
  $x\ge1$ was chosen arbitrarily, we can deduce $P_{i+1}(\overline c)$.

  Conversely suppose $P_{i+1}(\overline c)$. Let $T$ belong to
  $H^{i+1}_{\overline c}$. By the induction hypothesis, it is one of
  the trees $U'_{\kappa,x}$ for some $x\in \N_+$, $\kappa\in\N_+\cup\{\omega\}$. In
  any case, it also belongs to $J^{i+1}_\omega$. Hence it remains to
  show that any tree of the form $U'_{\omega,x}$ belongs to
  $H^{i+1}_{\overline c}$. So let $x\in \N_+$. Then, by $P_{i+1}(\overline
  c)$, there exists $y\in\N_+$ with $P_i(\overline c,x,y)$. By the
  induction hypothesis, we have $T^i_{\overline cxy}\cong U^i_\omega$
  and therefore $T'_{\overline cxy}\cong U'_{\omega,x}$ (which belongs
  to $H^{i+1}_{\overline c}$ by the very definition).  \qed
\end{proof}

\begin{lemma}\label{lem:tree_main1a}
  For all $\overline{c} \in \N_+^k$ there exists
  $\kappa \in \N_+\cup\{\omega\}$ such that
  $T^{i+1}_{\overline{c}} \cong  U^{i+1}_\kappa$.
\end{lemma}

\begin{proof}
It suffices to prove that
$H^{i+1}_{\overline{c}} \sim J^{i+1}_\kappa$ for some
$\kappa \in \N_+\cup\{\omega\}$.
Choose $\kappa$ as the smallest value in $\N_+\cup\{\omega\}$
such that
$$
\forall x \geq \kappa\, \forall y : \neg P_i(\overline{c},x,y)
$$
holds. By property (ii) from
Lemma~\ref{lem:normalform} for $P_i$, we get
$$
\forall 1 \leq x < \kappa \, \exists y : P_i(\overline{c},x,y).
$$
By the induction hypothesis, we get
$$
\forall x \geq \kappa\, \forall y : T'_{\overline{c}xy}
\not\cong U'_{\omega,x}
\ \text{ and } \
\forall 1 \leq x < \kappa \, \exists y :  T'_{\overline{c}xy} \cong
U'_{\omega,x}\ .
$$
It follows that $H^{i+1}_{\overline{c}}$ contains, apart from the
trees in $H^{i+1} = \biguplus\{ U'_{m,x} \mid x,m\in \N_+\}$, exactly
the trees from $\biguplus\{ U'_{\omega,x} \mid 1 \leq x < \kappa \}$. Hence,
$H^{i+1}_{\overline{c}} \sim J^{i+1}_\kappa$.
\qed
\end{proof}
Lemma~\ref{lem:tree_main1} and \ref{lem:tree_main1a} immediately imply:

\begin{lemma}\label{lem:tree_main2}
  For all $\overline{c} \in \N_+^k$ we have
  \[
  P_{i+1}(\overline{c}) \text{ does not hold }
   \iff \exists m \in \N_+ : T^{i+1}_{\overline{c}}\cong U^{i+1}_m.
  \]
\end{lemma}
In summary, we obtained the following:

\begin{proposition}\label{P-for-hardness}
  Let $n\ge2$ and let $P(x)$ be a $\Pi^0_{2n-3}$-predicate. Then, for
  any $c\in\N_+$, we have
  \[
    P(c)\text{ holds }\iff T^n_{c}\cong U^n_\omega\ .
  \]
\end{proposition}
To infer the $\Pi^0_{2n-3}$-hardness of the isomorphism problem for
$\T_n$ from this proposition, it remains to be shown that the trees
$T^n_{c}$ and $U^n_\omega$ are effectively automatic -- this is the
topic of the next section.

\subsection{Automaticity}

For constructing automatic presentations for the trees
from the previous section, it is actually easier to work
with {\em dags} ({\em directed acyclic graphs}).
The {\em height} of a dag $D$ is the length (number of edges)
of a longest directed path in $D$. We only consider dags
of finite height. A {\em root} of a dag is a node without incoming edges. A dag $D = (V,E)$ can be unfolded into
a forest $\unfold(D)$ in the usual way:
Nodes of $\unfold(D)$ are directed paths in $D$
that cannot be extended to the left (i.e., the initial node
of the path is a root) and there is an edge
between a path $p$ and a path $p'$ if and only if $p'$ extends
$p$ by one more node. For a node $v \in V$ of $D$, we define
the tree $\unfold(D,v)$ as follows: First we restrict
$D$ to those nodes that are reachable from $v$ and then
we unfold the resulting dag.
We need the following lemma.

\begin{lemma}\label{from dags to trees}
  From given $k \in \N$ and an automatic dag $D = (V,E)$ of height at
  most $k$, one can construct effectively an automatic
  presentation~$\P$ with $\S(\P)\cong\unfold(D)$.
\end{lemma}

\begin{proof}
  The universe for our automatic copy of $\unfold(D)$ is the
  set $P$ of all convolutions $v_1 \otimes v_2 \otimes \cdots \otimes
  v_m$, where $v_1$ is a root and $(v_i, v_{i+1}) \in E$ for all $1
  \leq i < m$. Since $D$ has height at most $k$, we have $m \leq k$.
  Since the edge relation of $D$ is automatic and since the set of all
  roots in $D$ is first-order definable and hence regular, $P$ is
  indeed a regular set.  Moreover, the edge relation of $\unfold(D)$
  becomes clearly FA recognizable on $P$.  \qed
\end{proof}
For $2\le i\le n$, let us consider the following forest:
\begin{align*}
  F^i =
    & \biguplus \{T^i_{\overline c}\mid \overline c\in\N_+^{1+2(n-i)}\}
     \uplus \biguplus \{U^i_\kappa \mid \kappa\in\N_+ \cup \{\omega\}\}\ .
\end{align*}
Technically, this section proves by induction over $i$ the following
statement:
\begin{proposition}\label{P:forest-automatic}
  For $2\le i\le n$, there exists an automatic copy $\mathcal F^i$ of
  $F^i$ and an isomorphism $f^i:F^i\to \mathcal F^i$ that maps
  \begin{enumerate}
  \item the root of the tree $T^i_{\overline c}$ to $a^{\overline c}$
    (for all $\overline c\in\N_+^{1+2(n-i)}$),
  \item the root of the tree $U^i_\omega$ to $\varepsilon$, and
  \item the root of the tree $U^i_m$ to $b^m$ (for all
    $m\in\N_+$).
  \end{enumerate}
\end{proposition}
This will give the desired result since $T^n_c$ is then
isomorphic to the connected component of $\mathcal F^n$ that contains
the word $a^c$ (and similarly for $U^n_\kappa$).  Note
that this connected component is automatic by
Theorem~\ref{thm:extFOaut}, since the forest $\mathcal F^n$ has
bounded height. Moreover, an automatic presentation for 
the connected component containing $a^c$ can be computed from $c$.

By Lemma~\ref{from dags to trees}, it suffices to construct an
automatic dag $\mathcal{D}^i$ such that there is an isomorphism
$h:\unfold(\mathcal{D}^i) \to \mathcal F^i$ that is the identity on
the set of roots of $\mathcal D^i$.

\subsubsection{Induction base: the automatic dag $\mathcal D^2$}

Recall the definitions of
$\Sigma_\ell^a$, $a^{\overline e}$,
and $\otimes_k(L)$ from Section~\ref{sec:equiv}.

\begin{lemma}\label{L1}
  From $\ell\in\N_+$, $q_1,q_2\in\N[x_1,\dots,x_\ell]$, and
  a symbol $a$, one can compute an automatic
  forest of height $1$ over an alphabet
  $\Sigma_\ell^a \uplus\Gamma$ such that
  \begin{itemize}
  \item the roots are the words from $\otimes_\ell(a^+)$,
  \item the leaves are words from $\Gamma^+$, and
  \item the tree rooted at $a^{\overline{e}}$ is
    isomorphic to $T[q_1(\overline e),q_2(\overline e)]$.
  \end{itemize}
\end{lemma}

\begin{proof}
  Set $p(x_1,\ldots,x_\ell)=C(q_1(x_1,\ldots,x_\ell),q_2(x_1,\ldots,x_\ell))$
  and recall the definition of the automata $\A[p]$ and $\Run_{\A[p]}$
  from Section~\ref{sec:equiv}.  Recall also that we let $\pi$ be the
  projection with $\pi(p,a,q)=a$ for a transition $(p,a,q)$ of $\A[p]$.
  Then let
  \begin{eqnarray*}
  L[q_1,q_2] & = & \otimes_\ell(a^+) \cup (\pi^{-1}(\otimes_\ell(a^+)) \cap L(\Run_{\A[p]}))\text{ and } \\
  E[q_1,q_2] & = & \{ (u,v) \mid u\in \otimes_\ell(a^+), v\in \pi^{-1}(u) \cap L(\Run_{\A[p]}) \} \ .
  \end{eqnarray*}
  Then $L[q_1,q_2]$ is regular and
  $E[q_1,q_2]$ is FA recognizable, i.e., the pair $(L[q_1,q_2];E[q_1,q_2])$
  is an automatic graph. It is actually a forest of height $1$, the
  words from $\otimes_\ell(a^+)$ form the roots, and the tree rooted
  at $a^{\overline{e}}$ has precisely
  $p(\overline{e})$ leaves, i.e., it is isomorphic to
  $T[q_1(\overline e),q_2(\overline e)]$.\qed
\end{proof}
{}From now on, we use the notations from Section~\ref{sss:base}.
Using Lemma~\ref{L1}, we can compute automatic forests $\mathcal F_1$
and $\mathcal F_2$ over alphabets
$\Sigma_{\ell+1}^a \uplus\Gamma_1$ and
$\Sigma_2^b \uplus\Gamma_2$, respectively, such that
\begin{enumerate}[(a)]
  \item the roots of $\mathcal F_1$ are the words from
    $\otimes_{\ell+1}(a^+)$,
  \item the roots of $\mathcal F_2$ are the words from $\otimes_2(b^+)$,
  \item the leaves of $\mathcal F_i$ are words from $\Gamma_i^+$ ($i \in \{1,2\}$),
  \item the tree rooted at $a^{\overline{e} e_{\ell+1}}$
    is isomorphic to $T[p_1(\overline e)+e_{\ell+1},p_2(\overline
    e)+e_{\ell+1}]$ for $\overline e \in\N_+^{\ell}$, $e_{\ell+1}\in\N_+$,
  \item the tree rooted at $b^{e_1 e_2}$ is isomorphic to
    $T[e_1,e_2]$ for $e_1,e_2\in\N_+$.
\end{enumerate}
We can assume that the alphabets $\Gamma_1$, $\Gamma_2$,
$\Sigma_{\ell+1}^a$, and $\Sigma_2^b$ are mutually disjoint.
Let $\mathcal F=(V_{\mathcal F},E_{\mathcal F})$ be the disjoint union of $\mathcal F_1$
and $\mathcal F_2$; it is effectively automatic.

The universe of the automatic dag $\mathcal D^2$ is the 
regular language 
$$
\otimes_k(a^+) \cup b^* \cup (\$^* \otimes V_{\mathcal F}),
$$
where $\$$ is a new symbol.
We have the following edges:
\begin{itemize}
\item For $u,v \in V_{\mathcal F}$, $\$^m \otimes u$ is connected to $\$^n
  \otimes v$ if and only if $m=n$ and $(u,v) \in E_{\mathcal F}$.
  This produces $\aleph_0$ many copies of $\mathcal{F}$.

\item $a^{\overline c}$ is connected to any word from $\$^* \otimes (\{ a^{\overline c\, \overline x} \mid \overline x \in
 \N_+^{\ell-k+1} \} \cup \{b^{e_1 e_2} \mid e_1 \neq e_2 \})$. By point (d) and (e) above, this means that
 the tree $\unfold(\mathcal D^2, a^{\overline c})$ has $\aleph_0$ many subtrees isomorphic to
 $T[p_1(\overline c \, \overline x)+x_{\ell+1},
    p_2(\overline c\, \overline x)+x_{\ell+1}]$ for $\overline x \in \N_+^{\ell-k}$, $x_{\ell+1}
    \in \N_+$ and $T[e_1,e_2]$ for $e_1,e_2 \in \N_+$, $e_1 \neq e_2$. 
 Hence, $\unfold(\mathcal D^2, a^{\overline c}) \cong T^2_{\overline{c}}$.

\item $\varepsilon$ is connected to all words from $\$^* \otimes \{b^{e_1 e_2} \mid e_1 \neq e_2 \}$. By (e) above, 
this means that the tree $\unfold(\mathcal D^2, \varepsilon)$ has $\aleph_0$ many subtrees 
isomorphic to $T[e_1,e_2]$ for $e_1,e_2 \in \N_+$, $e_1 \neq e_2$. Hence, $\unfold(\mathcal D^2, \varepsilon) \cong U^2_\omega$.

\item $b^m$ ($m \in \N_+$) is connected to all words from 
 $\$^* \otimes \{b^{e_1 e_2} \mid e_1 \neq e_2 \text { or } e_1=e_2 > m \}$. By (e) above,
 this means that the tree $\unfold(\mathcal D^2, b^m)$ has $\aleph_0$ many subtrees isomorphic to
 $T[e_1,e_2]$ for all $e_1,e_2 \in \N_+$ with $e_1 \neq e_2$ or 
 $e_1=e_2 > m$. Hence, $\unfold(\mathcal D^2, b^m) \cong U^2_m$.
\end{itemize}
Thus, $\unfold(\mathcal{D}_2) \cong F^2$ and the roots are as 
required in Proposition~\ref{P:forest-automatic}, see Figure~\ref{fig:auto_tree1}.
Moreover, it is clear that $\mathcal{D}_2$ is automatic.

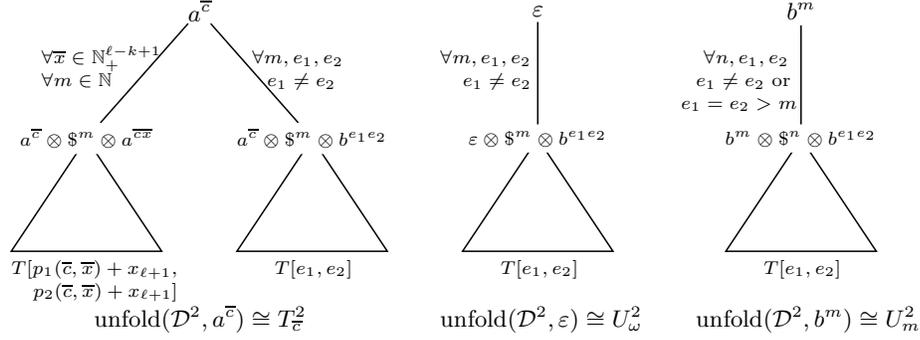
\begin{figure}[t]
\begin{center}
\setlength{\unitlength}{1mm}
\begin{picture}(20,40)(0,0)
\gasset{Nadjustdist=0.8,Nadjust=wh,Nframe=n,Nfill=n,AHnb=0,ELdist=0.5,linewidth=.2}
   \drawpolygon[Nframe=y,Nfill=n](-5,23)(-15,8)(5,8)
    \drawpolygon[Nframe=y,Nfill=n](-35,23)(-45,8)(-25,8)
   \put(-34,-2){$\unfold(\mathcal{D}^2,a^{\overline{c}}) \cong T^2_{\overline{c}}$}
   \node(root)(-20,40){$a^{\overline{c}}$}
   \node[fillcolor=white](a)(-35,23){\scriptsize $a^{\overline{c}}\otimes \$^m \otimes a^{\overline{c}\overline{x}}$}
    \node[fillcolor=white](a'')(-5,23){\scriptsize $a^{\overline{c}}\otimes \$^m\otimes b^{e_1e_2}$}
   \drawedge[ELside=r](root,a){}
   \put(-41,30){\text{\scriptsize $\forall m\in \N$}}
   \put(-41,33){\text{\scriptsize $\forall \overline{x}\in \N_+^{\ell-k+1}$}}
    \drawedge[ELside=l,ELpos=60,ELdist=0](root,a''){}
    \put(-13,33){\text{\scriptsize $\forall m, e_1,e_2$}}
    \put(-11,30){\text{\scriptsize $e_1\neq e_2$}}
\gasset{Nw=1,Nh=1,Nframe=y,Nfill=y,Nadjust=n}
   \put(-45,5){\text{\scriptsize $T[p_1(\overline{c},\overline{x})+x_{\ell+1},$}}
   \put(-42.1,2){\text{\scriptsize $p_2(\overline{c},\overline{x})+x_{\ell+1}]$}}
   \put(-10,5){\text{\scriptsize $T[e_1,e_2]$}}


\gasset{Nadjustdist=0.8,Nadjust=wh,Nframe=n,Nfill=n,AHnb=0,ELdist=0.5,linewidth=.2}
   \drawpolygon[Nframe=y,Nfill=n](25,23)(15,8)(35,8)
   \put(12,-2){$\unfold(\mathcal{D}^2, \varepsilon) \cong U^2_{\omega}$}
   \node(root)(25,40){$\varepsilon$}
   \node[fillcolor=white](a)(25,23){\text{\scriptsize $\varepsilon \otimes \$^m \otimes b^{e_1e_2}$}}
   \drawedge[ELside=r](root,a){}
   \put(12,33){\text{\scriptsize $\forall m, e_1,e_2$}}
   \put(15,30){\text{\scriptsize $e_1\neq e_2$}}
\gasset{Nw=1,Nh=1,Nframe=y,Nfill=y,Nadjust=n}
   \put(20,5){\text{\scriptsize $T[e_1,e_2]$}}


\gasset{Nadjustdist=0.8,Nadjust=wh,Nframe=n,Nfill=n,AHnb=0,ELdist=0.5,linewidth=.2}
    \drawpolygon[Nframe=y,Nfill=n](60,23)(50,8)(70,8)
   \put(46,-2){$\unfold(\mathcal{D}^2,b^m)\cong U^2_{m}$}
   \node(root)(60,40){$b^m$}
    \node[fillcolor=white](a'')(60,23){\scriptsize $b^m\otimes \$^n \otimes b^{e_1e_2}$}
    \drawedge[ELside=l,ELpos=60,ELdist=0](root,a''){}
    \put(47,33){\text{\scriptsize $\forall n,e_1,e_2$}}
    \put(46,30){\text{\scriptsize $e_1\neq e_2$ or}}
    \put(44,27){\text{\scriptsize $e_1=e_2>m$}}
\gasset{Nw=1,Nh=1,Nframe=y,Nfill=y,Nadjust=n}
   \put(55,5){\text{\scriptsize $T[e_1,e_2]$}}

\end{picture}
\end{center}
\caption{\label{fig:auto_tree1} Automatic presentation of $T^2_{\overline{c}}$ and $U^2_{\kappa}$}
\end{figure}

\subsubsection{Induction step: the automatic dag $\mathcal D^{i+1}$}
\label{sec-induction-automatic-trees}

Suppose $\mathcal D^i=(V,E)$ is such that $\mathcal F^i = \unfold(\mathcal
D^i)$ is as described in Proposition~\ref{P:forest-automatic}.

We use the notations from Section~\ref{sec-induction-trees}. We first
build another automatic dag $\mathcal D'$, whose unfolding will comprise (copies
of) all the trees $U'_{\kappa,x}$ ($\kappa \in \N_+\cup\{\omega\}$, $x \in
\N_+$) and $T'_{\overline cxy}$ ($\overline c \in \N_+^k$, $x,y \in \N_+$).
Recall that the set of roots of $\mathcal D^i$ is
$\otimes_{k+2}(a^+) \cup b^* \subseteq V$.
The universe of $\mathcal D'$ consists of the regular language
$$
(V \setminus b^*) \cup  (\sharp^+ \otimes b^*) \cup \sharp_1^+ \sharp_2^*,
$$
where $\sharp, \sharp_1$, and $\sharp_2$ are new symbols.
We have the following edges in $\mathcal D'$:
\begin{itemize}
\item All edges from $E$ except those with an initial node in $b^*$
are present in $\mathcal D'$.
\item $a^{\overline c xy} \in V$ is connected to all words of the form
$\sharp_1^i \sharp_2^{x-i}$ for $\overline c \in \N_+^k,x,y \in \N_+$, and $1 \leq i \leq x$.
This ensures that the subtree rooted at $a^{\overline c xy}$ gets $x$ new leaves,
which are children of the root. Hence $\unfold(\mathcal D',a^{\overline c xy})
\cong T'_{\overline cxy}$.
\item $\sharp^x \otimes b^m$ for $x \in \N_+$ and $m \in \N$
 is connected  to (i) all nodes to which $b^m$ is connected
in $\mathcal D^i$ and  to (ii) all nodes from $\sharp_1^i \sharp_2^{x-i}$ for $1 \leq i \leq x$.
This ensures that $\unfold(\mathcal D',\sharp^x \otimes b^m)
 \cong U'_{m,x}$ in case $m \in \N_+$ and $\unfold(\mathcal D',\sharp^x \otimes \varepsilon)
 \cong U'_{\omega,x}$.
\end{itemize}
In summary, $\mathcal D'$ is a dag, whose unfolding
consists of (a copy
of) $U'_{\omega,x}$ rooted at $\sharp^x\otimes\varepsilon$, $U'_{m,x}$ ($m \in \N_+$)
rooted at $\sharp^x \otimes b^m$, and $T'_{\overline cxy}$ rooted at
$a^{\overline cxy}$.

{}From the automatic dag $\mathcal D'$, we now build in a final step
the automatic dag $\mathcal D^{i+1}$. This is very similar to the
constructions of $\mathcal D^2$ and $\mathcal D'$ above.
Let $V'$ be the universe of $\mathcal D'$.
The universe of $\mathcal D^{i+1}$ is the regular language
$$
\otimes_k(a^+) \cup b^* \cup (\$^* \otimes V')\ .
$$
The edges are as follows:
\begin{itemize}
\item For $u,v \in V'$, $\$^m \otimes u$ is connected to $\$^n \otimes v$ 
 if and only if $m=n$ and $(u,v)$ is an edge of $\mathcal D'$.
 This generates $\aleph_0$ many copies of $\mathcal D'$.
\item $a^{\overline c}$ is connected to every word from 
 $\$^* \otimes (\{ a^{\overline c x y} \mid x,y \in \N_+ \} \cup (\sharp^+ \otimes b^+))$.
Hence, the tree $\unfold(\mathcal D^{i+1}, a^{\overline c})$ 
has $\aleph_0$ many subtrees isomorphic to $T'_{\overline cxy}$ for $x,y \in \N_+$ and $U'_{m,x}$
for $x,m \in \N_+$. Thus, $\unfold(\mathcal D^{i+1}, a^{\overline c}) \cong T^{i+1}_{\overline{c}}$.

\item $\varepsilon$ is connected to all words from 
$\$^* \otimes (\sharp^+ \otimes b^*)$. Hence, the tree $\unfold(\mathcal D^{i+1}, \varepsilon)$ has
$\aleph_0$ many subtrees isomorphic to 
$U'_{\kappa,x}$ for all $x \in \N_+$ and $\kappa \in \N_+ \cup \{\omega\}$. 
Thus, $\unfold(\mathcal D^{i+1}, \varepsilon) \cong U^{i+1}_{\omega}$.

\item $b^m$ ($m \in \N_+$) is connected to all words from $\$^* \otimes
  ((\sharp^+ \otimes b^+) \cup \{\sharp^x \otimes\varepsilon \mid 1 \leq x < m\})$. 
  This means that the tree $\unfold(\mathcal D^{i+1}, b^m)$ has 
  $\aleph_0$ many subtrees isomorphic to $U'_{m,x}$ for all
  $m,x \in \N_+$ and $U'_{\omega,x}$ for all $1 \leq x < m$. Hence, 
  $\unfold(\mathcal D^{i+1}, b^m) \cong U^{i+1}_m$.
\end{itemize}
See Figure~\ref{fig:auto_tree2}, \ref{fig:auto_tree3}, and
\ref{fig:auto_tree4} for the overall construction.
This finishes the proof of Proposition~\ref{P:forest-automatic}. Hence we
obtain:

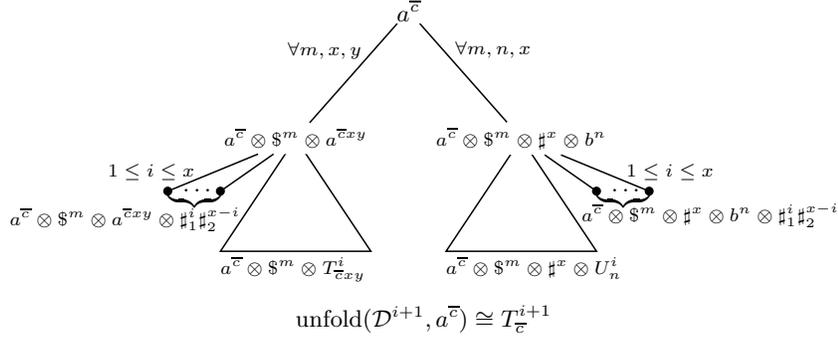
\begin{figure}[t]
\begin{center}
\setlength{\unitlength}{1mm}
\begin{picture}(0,40)(-20,0)
\gasset{Nadjustdist=0.8,Nadjust=wh,Nframe=n,Nfill=n,AHnb=0,ELdist=0.5,linewidth=.2}
  \drawpolygon[Nframe=y,Nfill=n](-40,23)(-50,8)(-30,8)
  \drawpolygon[Nframe=y,Nfill=n](-10,23)(-20,8)(0,8)
   \put(-40,-2){$\unfold(\mathcal{D}^{i+1},a^{\overline{c}}) \cong T^{i+1}_{\overline{c}}$}
   \node(root)(-25,40){$a^{\overline{c}}$}
   \node[fillcolor=white](a)(-40,23){\scriptsize $a^{\overline{c}}\otimes \$^m\otimes a^{\overline{c}xy}$}
    \node[fillcolor=white](a'')(-10,23){\scriptsize $a^{\overline{c}}\otimes \$^m \otimes \sharp^x \otimes b^n$}
   \drawedge[ELside=r](root,a){\scriptsize $\forall m,x,y$}
    \drawedge[ELside=l](root,a''){\scriptsize $\forall m,n,x$}
\gasset{Nw=1,Nh=1,Nframe=y,Nfill=y,Nadjust=n}
  
   \node(c)(-57,16){}
   \node(c')(-50,16){}
   \drawedge(a,c){}
   \put(-65,18){\scriptsize $1\leq i\leq x$}
   \drawedge(a,c'){}
   \put(-55,16){$\ldots$}
   \put(-57,16){$\underbrace{\makebox(7,0){}}_{}$}
   \put(-78,11.5){\scriptsize $a^{\overline{c}}\otimes \$^m\otimes a^{\overline{c}xy} \otimes \sharp_1^{i}\sharp_2^{x-i}$}
   \put(-50,5){\text{\scriptsize $a^{\overline{c}}\otimes \$^m\otimes T^i_{\overline{c}xy}$}}
   \node(d)(0,16){}
   \node(d')(7,16){}
   \drawedge(a'',d){}
   \drawedge(a'',d'){}
   \put(4,18){\scriptsize $1\leq i\leq x$}
   \put(2,16){$\ldots$}
   \put(0,16){$\underbrace{\makebox(7,0){}}_{}$}
   \put(-2,12){\scriptsize $a^{\overline{c}}\otimes \$^m\otimes \sharp^x \otimes b^n \otimes \sharp_1^{i}\sharp_2^{x-i}$}

   \put(-20,5){\text{\scriptsize $a^{\overline{c}}\otimes \$^m\otimes \sharp^x\otimes U^i_{n}$}}
\end{picture}
\end{center}
\caption{\label{fig:auto_tree2} Automatic presentation of $T^{i+1}_{\overline{c}}$}
\end{figure}

\begin{figure}[t]
\begin{center}
\setlength{\unitlength}{1mm}
\begin{picture}(0,40)(-20,0)
\gasset{Nadjustdist=0.8,Nadjust=wh,Nframe=n,Nfill=n,AHnb=0,ELdist=0.5,linewidth=.2}
    \drawpolygon[Nframe=y,Nfill=n](-10,23)(-20,8)(0,8)
    \drawpolygon[Nframe=y,Nfill=n](-40,23)(-50,8)(-30,8)
   \put(-40,-2){$\unfold(\mathcal{D}^{i+1},\varepsilon) \cong U^{i+1}_{\omega}$}
   \node(root)(-25,40){$\varepsilon$}
   \node[fillcolor=white](a)(-40,23){\scriptsize $\varepsilon\otimes \$^m\otimes \sharp^x \otimes \varepsilon$}
    \node[fillcolor=white](a'')(-10,23){\scriptsize $\varepsilon \otimes \$^m \otimes \sharp^x \otimes b^n$}
   \drawedge[ELside=r](root,a){\scriptsize $\forall m,x$}
    \drawedge[ELside=l](root,a''){\scriptsize $\forall m,n,x$}
\gasset{Nw=1,Nh=1,Nframe=y,Nfill=y,Nadjust=n}
   
   \node(c)(-57,16){}
   \node(c')(-50,16){}
   \drawedge(a,c){}
   \put(-65,18){\scriptsize $1\leq i\leq x$}
   \drawedge(a,c'){}
   \put(-55,16){$\ldots$}
   \put(-57,16){$\underbrace{\makebox(7,0){}}_{}$}
   \put(-78,11.5){\scriptsize $\varepsilon\otimes \$^m\otimes \sharp^x\otimes \varepsilon \otimes \sharp_1^{i}\sharp_2^{x-i}$}
   \put(-50,5){\text{\scriptsize $\varepsilon \otimes \$^m\otimes \sharp^x \otimes U^i_{\omega}$}}
   \node(d)(0,16){}
   \node(d')(7,16){}
   \drawedge(a'',d){}
   \drawedge(a'',d'){}
   \put(4,18){\scriptsize $1\leq i\leq x$}
   \put(2,16){$\ldots$}
   \put(0,16){$\underbrace{\makebox(7,0){}}_{}$}
   \put(-2,12){\scriptsize $\varepsilon \otimes \$^m\otimes \sharp^x \otimes b^n \otimes \sharp_1^{i}\sharp_2^{x-i}$}

   \put(-20,5){\text{\scriptsize $\varepsilon\otimes \$^m\otimes \sharp^x\otimes U^i_{n}$}}

\end{picture}
\end{center}
\caption{\label{fig:auto_tree3} Automatic presentation of $U^{i+1}_{\omega}$}
\end{figure}
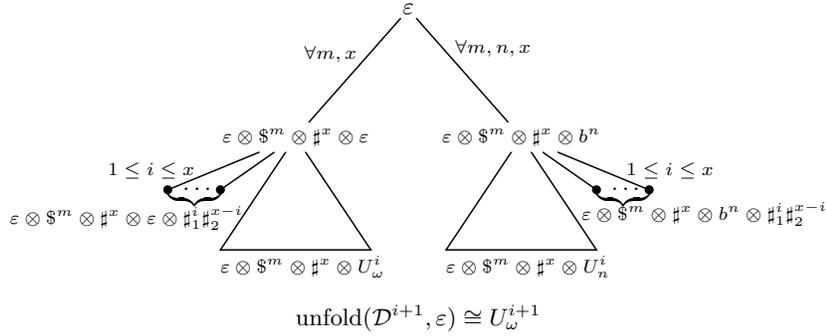

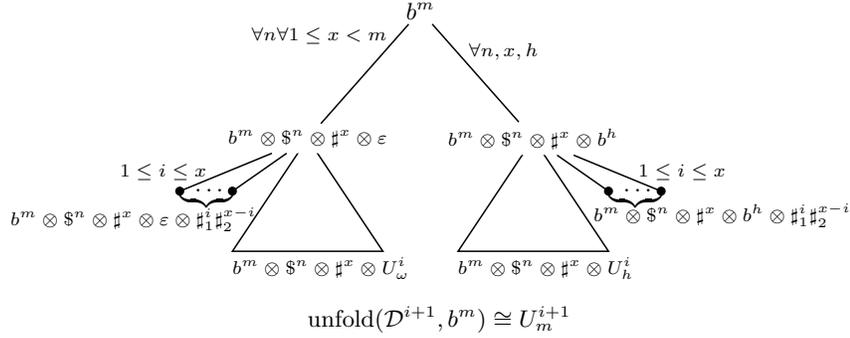
\begin{figure}[t]
\begin{center}
\setlength{\unitlength}{1mm}
\begin{picture}(0,40)(-20,0)
 \gasset{Nadjustdist=0.8,Nadjust=wh,Nframe=n,Nfill=n,AHnb=0,ELdist=0.5,linewidth=.2}
  \drawpolygon[Nframe=y,Nfill=n](-40,23)(-50,8)(-30,8)
   \drawpolygon[Nframe=y,Nfill=n](-10,23)(-20,8)(0,8)
   \put(-40,-2){$\unfold(\mathcal{D}^{i+1},b^m) \cong U^{i+1}_{m}$}
   \node(root)(-25,40){$b^m$}
   \node[fillcolor=white](a)(-40,23){\scriptsize $b^m\otimes \$^n\otimes \sharp^x \otimes \varepsilon$}
    \node[fillcolor=white](a'')(-10,23){\scriptsize $b^m \otimes \$^n \otimes \sharp^x \otimes b^h$}
   \drawedge[ELside=r](root,a){\scriptsize $\forall n \forall 1\leq x<m$}
    \drawedge[ELside=l](root,a''){\scriptsize $\forall n,x,h$}
  \gasset{Nw=1,Nh=1,Nframe=y,Nfill=y,Nadjust=n}
   
   \node(c)(-57,16){}
   \node(c')(-50,16){}
   \drawedge(a,c){}
   \put(-65,18){\scriptsize $1\leq i\leq x$}
   \drawedge(a,c'){}
   \put(-55,16){$\ldots$}
   \put(-57,16){$\underbrace{\makebox(7,0){}}_{}$}
   \put(-79.5,11.5){\scriptsize $b^m\otimes \$^n\otimes \sharp^x\otimes \varepsilon \otimes \sharp_1^{i}\sharp_2^{x-i}$}
   \put(-50,5){\text{\scriptsize $b^m \otimes \$^n\otimes \sharp^x\otimes  U^i_{\omega}$}}
   \node(d)(0,16){}
   \node(d')(7,16){}
   \drawedge(a'',d){}
   \drawedge(a'',d'){}
   \put(4,18){\scriptsize $1\leq i\leq x$}
   \put(2,16){$\ldots$}
   \put(0,16){$\underbrace{\makebox(7,0){}}_{}$}
   \put(-2,12){\scriptsize $b^m \otimes \$^n\otimes \sharp^x \otimes b^h \otimes \sharp_1^{i}\sharp_2^{x-i}$}
   \put(-20,5){\text{\scriptsize $b^m\otimes \$^n\otimes \sharp^x\otimes U^i_{h}$}}
\end{picture}
\end{center}
\caption{\label{fig:auto_tree4} Automatic presentation of $U^{i+1}_{m}$}
\end{figure}

\begin{theorem}\label{thm:tree}
  \begin{enumerate}
  \item For any $n\ge2$, the isomorphism problem for automatic trees
    of height at most $n$ is $\Pi^0_{2n-3}$-complete.
  \item The isomorphism problem for the class of automatic trees of
    finite height is recursively equivalent to $\FOTh(\N; +, \times)$.
  \end{enumerate}
\end{theorem}

\begin{proof}
  We first prove the first statement. Containment in $\Pi^0_{2n-3}$
  was shown in Proposition~\ref{prop:tree_membership}. For the hardness, let
  $P_n\subseteq\N_+$ be any $\Pi^0_{2n-3}$-predicate and let
  $c\in\N_+$. Then, above, we constructed the automatic forest
  $\mathcal F^n$ of height $n$. The trees $T^n_c$ and $U^n_\omega$ are first-order
  definable in $\mathcal F^n$ since they are (isomorphic to) the trees
  rooted at $a^c$ and $\varepsilon$, resp. Hence these two
  trees are automatic. By Proposition~\ref{P-for-hardness}, they are
  isomorphic if and only if $P_n(c)$ holds.

  We now come to the second statement. Since the proof of 
  Prop.~\ref{prop:tree_membership} is uniform in the level~$n$, we can compute
  from two automatic trees $T_1, T_2$ of finite height an arithmetical
  formula, which is true if and only if $T_1 \cong T_2$. For the other
  direction, one observes that the height of an
  automatic tree of finite height can be computed. Then the result
  follows from the first statement because of the uniformity of its
  proof.\qed
\end{proof}
In fact, we proved a slightly stronger statement:
For every $n \geq 2$, there exists a fixed $\Pi^0_{2n-3}$-complete
set $P_{2n-3} \subseteq \N_+$. If we apply our construction, we obtain a {\em
  fixed automatic forest} $\F^n$
of height $n$ 
with the following properties: It is $\Pi^0_{2n-3}$-complete to determine, whether 
for given $c \in \N_+$, the tree rooted at $a^c$ in $\F^n$ is isomorphic to the 
tree rooted at $\varepsilon$ in $\F^n$.


\section{Recursive trees of finite height}

In this section, we briefly discuss the isomorphism problem for
recursive (i.e., computable) trees of finite height.

\begin{theorem}
  For every $n \geq 1$, the isomorphism problem for recursive trees of
  height at most $n$ is $\Pi^0_{2n}$-complete.
\end{theorem}

\begin{proof}
  For the upper bound, let us first assume that $n=1$.  Two recursive
  trees $T_1$ and $T_2$ of height $1$ are
  isomorphic if and only if: for every $k \geq 0$, there exist at
  least $k$ nodes in $T_1$ if and only if there exist at least $k$
  nodes in $T_2$. This is a $\Pi^0_2$-statement.  For the inductive
  step, we can reuse the arguments from the proof of
  Proposition~\ref{prop:tree_membership}.

  For the lower bound, we first note that the isomorphism problem for
  recursive trees of height $1$ is $\Pi^0_2$-complete. It is known
  that the problem whether a given recursively enumerable set is
  infinite is $\Pi^0_2$-complete \cite{Rogers}. For a given
  deterministic Turing-machine $M$, we construct a recursive tree
  $T(M)$ of height $1$ as follows: the set of leaves of $T(M)$ is the
  set of all accepting computations of $M$. We add a root to the tree
  and connect the root to all leaves. If $L(M)$ is infinite, then  
  $T(M)$ is isomorphic to the height-$1$ tree with infinitely many
  leaves. If $L(M)$ is finite, then there exists $m \in \N$ such that
  $T(M)$ is isomorphic to the height-$1$ tree with $m$ leaves.  We can
  use this construction as the base case for our construction in
  Section~\ref{sec-induction-trees}. This yields the lower bound for
  all $n \geq 1$.  \qed
\end{proof}

\section{Automatic Linear Orders} \label{sec-lin-order}

We use $\omega$ to denote the linear order (type of) $(\N;\leq)$ of the natural
numbers and $\mathbf{n}$ to denote the finite linear order (type) of
size $n$. Let $I = (D_I;\leq_I)$ be a linear order and let
$\L=\{L_i\mid i\in D_I\}$ be a class of linear orders, where
$L_i=(D_i;\leq_i)$ for $i\in D_I$. The {\em sum} $\sum \L$ is the
linear order $( \{ (x,i) \mid i\in D_I, x \in D_i\};\leq)$ where
for all $i,j\in D_I$, $x\in D_i$, and $y\in D_j$,
\[
    (x,i)\leq (y,j) \ \Longleftrightarrow \ i <_I j\vee (i=j \wedge x\leq_i y)\ .
\]
We use $L_1+L_2$ to denote $\sum \{L_i\mid i\in \mathbf{2}\}$. We denote with $L_1\cdot L_2$ 
the sum $\sum \{L^i_1\mid i\in L_2\}$ where $L^i_1\cong L_1$ for every $i\in L_2$. 
An {\em interval} of a linear order $L= (D;\leq)$ is a subset 
$I\subseteq D$ such that $x,y\in I$ and $x< z< y$ imply $z\in I$.

A well-known example of an automatic linear order is the {\em
  lexicographic order} $\leq_{\lex}$ on a regular language $D$. To
define $\leq_{\lex}$, we first need a fixed linear order $<$ on the
alphabet $\Sigma$ of $D$. For $w,w'\in D$, we say that $w$ is {\em
  lexicographically less} than $w'$, denoted by $w<_{\lex} w'$, if
either $w$ is a proper prefix of $w'$ or there exist $x,y,z\in
\Sigma^*$ and $\sigma, \tau\in \Sigma$ such that $w=x\sigma y$,
$w'=x\tau z$ and $\sigma< \tau$. We write $w\leq_{\lex} w'$ if either
$w=w'$ or $w<_{\lex} w'$.  For convenience, in this paper, we use
$\leq_{\lex}$ to denote the lexicographic order regardless of the
corresponding alphabets and orders on the alphabets.  The precise
definition of $\leq_{\lex}$ in different occurrences will be clear
from the context.

This section is devoted to proving that the isomorphism problem on the
class of automatic linear orders is at least as hard as $\FOTh(\N;+,\times)$.
To this end, it suffices to prove (uniformly in $n$) $\Sigma^0_n$-hardness for
every even $n$. The general plan for this is similar to the proof for trees
of finite height: we use Hilbert's $10^{th}$ problem to handle
$\Pi^0_1$-predicates in several variables and an inductive
construction of more complicated linear orders to handle quantifiers,
i.e., to proceed from a $\Pi^0_{2i-1}$- to a $\Sigma^0_{2i}$-predicate
(and from a $\Sigma^0_{2i}$- to a $\Pi^0_{2i+1}$-predicate).

So let $n\ge1$ be even and let $P_n(x_0)$ be a
$\Sigma^0_n$-predicate. For every odd (even) number
$1\leq i< n$, let $P_i(x_0, \ldots, x_{n-i})$ be the
$\Pi^0_i$-predicate ($\Sigma^0_i$-predicate) such that $P_{i+1}(x_0,
\ldots, x_{n-i-1})$ is logically equivalent to $Q x_{n-i}:
P_i(x_0,\ldots, x_{n-i})$ where $Q=\exists$ if $i$ is odd and
$Q=\forall$ if $i$ is even.  We fix these predicates for the rest of
Section~\ref{sec-lin-order}.

By induction on $1\leq i\leq n$, we will construct
from $\overline{c}\in \N_+^{n-i+1}$ the following linear orders:
\begin{itemize}
\item a test linear order $L^i_{\overline{c}}$,
\item a linear order $K^i$, and
\item a set of linear orders $\M^i$ such that $\M^1=\{M^1_m\mid m\in \N_+\}$ and $\M^i$ is the singleton $\{M^i\}$ if $i>1$.
\end{itemize}
These linear orders will have the following properties:
\begin{description}
\item[(P1)] $P_i(\overline c)$ holds if and only if
  $L^i_{\overline c} \cong K^i$.
\item[(P2)] $P_i(\overline c)$ does not hold if and only if $L^i_{\overline{c}}\cong M$ for some $M\in \M^i$.
\item[(P3)] The linear order $\omega\cdot \mathbf{i}$ is not isomorphic to any interval of $L^i_{\overline{c}}, K^i, M$ where $M\in \M^i$.
\end{description}
In the rest of the section, we will  inductively construct $L^i_{\overline{c}}$, $K^i$, and $\M^i$  
and prove (P1), (P2), and (P3). The subsequent section is devoted to proving the effective automaticity of these linear orders.

\subsection{Construction of linear orders}

Our construction of linear orders is quite similar to the construction
for trees from Section~\ref{Construction of trees}.  One of the main
differences is that in the inductive step for trees, we went from a
$\Pi^0_i$-predicate directly to a $\Pi^0_{i+2}$-predicate. Thereby the
height of the trees only increased by one.  This was crucial in order
to get $\Pi^0_{2n-3}$-completeness for the isomorphism problem for
automatic trees of height $n \geq 2$. For automatic linear orders, we
split the construction into two inductive steps: in the first step, we
go from a $\Pi^0_i$-predicate ($i$ odd) to a
$\Sigma^0_{i+1}$-predicate, whereas in the second step, we go from a
$\Sigma^0_{i+1}$-predicate to a $\Pi^0_{i+2}$-predicate.

A key technique used in the construction is the shuffle sum of a class of linear orders. 
Let $I$ be a countable set. A {\em dense $I$-coloring} of $\Q$ is a mapping $c:\Q\rightarrow I$ 
such that for all $x,y\in \Q$ with $x<y$ and all $i\in I$ there exists $x<z<y$ with $c(z)=i$.

\begin{definition}
Let $\L=\{L_i\mid i\in I\}$ be a set of linear orders with $I$ countable and let $c:\Q\rightarrow I$ 
be a dense $I$-coloring of $\Q$. The shuffle sum of $\L$, denoted $\Shuf(\L)$, is the linear order $\sum_{x\in \Q} L_{c(x)}$.
\end{definition}
In the above definition, the isomorphism type of $\sum_{x\in \Q}
L_{c(x)}$ does not depend on the choice of the dense $I$-coloring $c$,
see e.g. \cite{Ros82}. Hence $\Shuf(\L)$ is indeed uniquely defined.

In this section, we will consider classes $\L_1$ and $\L_2$ of linear
orders that we consider as classes of isomorphism types. Therefore, we
use the following abbreviations:
\begin{itemize}
\item ``$L\in\L_1$'' denotes that $\L_1$ contains a linear order
  isomorphic to $L$,
\item ``$\L_1\subseteq\L_2$'' denotes $\forall L_1\in\L_1\;\exists
  L_2\in\L_2:L_1\cong L_2$, and
\item ``$\L_1=\L_2$'' abbreviates $\L_1\subseteq\L_2\subseteq\L_1$.
\end{itemize}

\subsubsection{Induction base: construction of $L^1_{\overline{c}}$, $K^1$, and $M^1_m$}
\label{construction-lo-ind-base}

Recall from Section~\ref{sec:equiv} that the polynomial function $C(x,y)=(x+y)^2+3x+y$
is injective. 
For $n_1, n_2 \in \N_+$, let $L[n_1, n_2]$ be the finite linear order of length $C(n_1, n_2)$. 

By applying Matiyasevich's theorem, we obtain two polynomials 
$p_1(\overline{x}), p_2(\overline{x})\in \N[\overline{x}]$ in $\ell$ variables, 
$\ell>n$, such that for all $\overline{c}\in \N_+^{n}$, the $\Pi^0_1$-predicate $P_1(\overline{c})$ holds if and only if
\[
\forall \overline{x}\in \N^{\ell-n}: \ p_1(\overline{c},\overline{x}) \neq p_2(\overline{c},\overline{x})\ .
\]
Fix $\overline{c}\in \N^n_+$ and $m\in \N_+$. We define the following four classes of finite linear orders:
\begin{align}
    \L^1_1(\overline{c}) &= \{L[p_1(\overline{c}, \overline{x})+x_{\ell+1}, 
   p_2(\overline{c},\overline{x})+x_{\ell+1}]\mid \overline{x}\in
   \N^{\ell-n}_+, x_{\ell+1}\in \N_+\} \label{class-L_1}\\
    \L^1_2(m) &= \{L[x+m,x+m]\mid x\in \N_+\} \label{class-L_2} \\
    \L^1_3 &= \{L[x+y,x]\mid x,y\in \N_+\} \label{class-L_3}\\
    \L^1_4 &= \{L[x,x+y]\mid x,y\in \N_+\} \label{class-L_4}
\end{align}
The linear orders $L^1_{\overline{c}}$, $K^1$, and $M^1_m$ are 
obtained by taking the shuffle sums of unions of the above classes of linear orders:
\begin{align*}
    L^1_{\overline{c}} = \Shuf (\L^1_1(\overline{c}) \cup \L^1_3 \cup \L^1_4), \qquad
    K^1 = \Shuf (\L^1_3 \cup \L^1_4), \qquad M^1_m = \Shuf (\L^1_2(m) \cup \L^1_3 \cup \L^1_4).
\end{align*}
The next lemma is needed to prove (P1) and (P2) for the
$\Pi^0_1$-predicate $P_1$.

\begin{lemma}\label{lem:shuf_finite}
  Suppose $\L_1$ and $\L_2$ are two countable sets of finite linear
  orders. Then
  \[
      \L_1=\L_2 \ \iff \ \Shuf(\L_1) \cong \Shuf(\L_2)
  \]
  and no interval of $\Shuf(\L_1)$ is isomorphic to $\omega$.
\end{lemma}
\begin{proof}
  If $\L_1=\L_2$, then it is clear that $\Shuf(\L_1) \cong
  \Shuf(\L_2)$. Conversely, suppose there exists an isomorphism $f$
  from $\Shuf(\L_1)$ to $\Shuf(\L_2)$. We prove below that
  $\L_1=\L_2$. By symmetry we only need to prove $\L_1 \subseteq
  \L_2$.

  Note that for $i\in \{1,2\}$, $\Shuf(\L_i)$ is obtained by replacing
  each $q\in \Q$ with some linear order $L^i_q$ (whose type is)
  contained in $\L_i$. For every $q\in \Q$, if $f(L^1_q)$ contains
  elements from $L^2_p$ and $L^2_{p'}$ for some $p<p'$, then
  $f(L^1_q)$ is infinite which is impossible. Therefore $f$ maps
  $L^1_q$ into $L^2_p$ for some $p\in \Q$.  Using the same argument
  with $f$ replaced by $f^{-1}$, we can also prove that $f^{-1}$ maps
  $L^2_p$ into $L^1_q$. Hence $L^1_q \cong L^2_p$. This means that for
  all $L\in \L_1$, there is $L'\in \L_2$ such that $L\cong
  L'$. Therefore $\L_1 \subseteq \L_2$. 

  If $x_1<x_2<\cdots$ in $\Shuf(\L_1)$, then there are $p<p'$ in $\Q$
  and $k<\ell$ in $\N_+$ such that $x_k\in L^1_p$ and $x_\ell\in
  L^1_{p'}$. But then the interval $[x_k,x_\ell]$ is infinite. Hence
  no interval in $\Shuf(\L_1)$ is isomorphic to $\omega$.\qed
\end{proof}
The next lemma states (P1) and (P2) for $i=1$:

\begin{lemma}
For any $\overline{c}\in \N^n_+$, we have:
\begin{enumerate}[(1)]
\item $P_1(\overline{c}) \text{ holds } \ \Longleftrightarrow \ L^1_{\overline{c}} \cong K^1$.
\item $P_1(\overline{c}) \text{ does not hold } \ 
\Longleftrightarrow \ \exists m\in \N_+:  L^1_{\overline{c}} \cong M^1_m.$
\end{enumerate}
\end{lemma}
\begin{proof}
For (1), we have 
\begin{eqnarray*}
P_1(\overline{c}) \ & \Longleftrightarrow & \  
\forall \overline{x}\in \N^{\ell-n}_+ : p_1(\overline{c},\overline{x}) \neq p_2(\overline{c},\overline{x}) \\
&  \Longleftrightarrow & \  \forall \overline{x}\in \N^{\ell-n}_+, x_{\ell+1}\in \N_+ :  p_1(\overline{c},\overline{x})+x_{\ell+1}\neq p_2(\overline{c},\overline{x})+x_{\ell+1} \\
 & \Longleftrightarrow & \  \forall \overline{x}\in \N^{\ell-n}_+, x_{\ell+1}\in \N_+ :  L[p_1(\overline{c},\overline{x})+x_{\ell+1}, p_2(\overline{c},\overline{x})+x_{\ell+1}] \in \L^1_3\cup \L^1_4 \\
 & \Longleftrightarrow & \   \L^1_1(\overline{c})\cup \L^1_3 \cup \L^1_4 =
 \L^1_3\cup \L^1_4\\
& \stackrel{\text{Lemma~\ref{lem:shuf_finite}}}{\Longleftrightarrow} & \ L^1_{\overline{c}} \cong K^1 .
\end{eqnarray*}
For (2), we get
\begin{eqnarray*}
\neg P_1(\overline{c}) \ & \Longleftrightarrow & \  \exists \overline{x}\in \N^{\ell-n}_+ : p_1(\overline{c},\overline{x}) = p_2(\overline{c},\overline{x}) \\
&  \Longleftrightarrow & \  \exists m \in \N_+ :  L[m+1,m+1] \in  \L^1_1(\overline{c}) \\
 & \Longleftrightarrow & \  \exists m \in \N_+ : (\forall k > m:  L[k,k] \in
 \L^1_1(\overline{c}) \wedge \forall 1 \leq k \leq m : L[k,k] \not\in
 \L^1_1(\overline{c})) \\
 & \Longleftrightarrow & \  \exists m\in\N_+: \L^1_1(\overline{c}) \cup \L^1_3 \cup \L^1_4 = \L^1_2(m) \cup \L^1_3\cup \L^1_4 \\
 & \stackrel{\text{Lemma~\ref{lem:shuf_finite}}}{\Longleftrightarrow} & \   \exists m\in \N_+:  L^1_{\overline{c}} \cong M^1_m\ .
\end{eqnarray*}
\qed
\end{proof}
Since $L_{\overline c}^1$, $K^1$, and $M^1_m$ are shuffle sums, they satisfy
(P3) by Lemma~\ref{lem:shuf_finite}.
This finishes the construction for the base case.

\subsubsection{First induction step: from $P_i$ to $P_{i+1}$ for $i$ odd} 
\label{first-induction-step}

Suppose $i\geq 1$ is an odd number.
For notational simplicity, we write $k$ for $n-i$. Thus, $P_{i+1}$ is a $k$-ary predicate and $P_i$ is a $(k+1)$-ary one. 
For all $\overline{c}\in \N^k_+$, $P_{i+1}(\overline{c})$ is logically equivalent to $\exists x: P_i(\overline{c},x)$.
Applying the inductive hypothesis, for any $\overline{c}\in \N^{k}_+$ and $x\in \N_+$, 
we obtain linear orders $L^i_{\overline{c}x}$, $K^i$, and the set $\M^i$ such that
\begin{itemize}
\item $P_i(\overline{c},x)$ holds if and only if $L^i_{\overline{c}x}\cong K^i$,
\item $P_i(\overline{c},x)$ does not hold if and only if $L^i_{\overline{c}x}
  \cong M$ for some $M\in \M^i$, and
\item $\omega\cdot \mathbf{i}$ is not isomorphic to any interval of $L^i_{\overline{c}x}$, $K^i$, or $M$ where $M\in \M^i$.
\end{itemize}
Fix $\overline{c}\in \N^k_+$. We define the following classes of linear orders:
\begin{equation}\label{eqt:lo_L_i+1}
    \L^{i+1}_1(\overline{c}) = \{ \omega\cdot \mathbf{i} + L^i_{\overline{c}x}\mid x\in \N_+ \}, \qquad
    \L^{i+1}_2 = \{ \omega\cdot \mathbf{i} + M\mid M\in \M^i\},  \qquad
    \L^{i+1}_3 = \{\omega\cdot \mathbf{i} + K^i\} .
\end{equation}
The linear orders $L^{i+1}_{\overline{c}}$, $K^{i+1}$, and $M^{i+1}$ 
are defined as shuffle sums of unions of the above classes of linear orders:
\begin{equation} \label{order-for-i+1}
    L^{i+1}_{\overline{c}} = \Shuf(\L^{i+1}_1(\overline{c})\cup \L^{i+1}_2), 
\qquad  K^{i+1} = \Shuf( \L^{i+1}_2 \cup \L^{i+1}_3), \qquad M^{i+1} = \Shuf(\L^{i+1}_2).
\end{equation}
Recall that the set $\M^i$ is a singleton for $i > 1$, consisting of
$M^i$.  The next lemma can be proved similarly as
Lemma~\ref{lem:shuf_finite}.

\begin{lemma}\label{lem:shuf_inf}
  Suppose $\L_1$ and $\L_2$ are two countable classes of linear orders
  such that each $L\in \L_1\cup \L_2$ is isomorphic to a linear order
  of the form $\omega\cdot \mathbf{i} + K$, where $\omega\cdot
  \mathbf{i}$ is not isomorphic to any interval of $K$. Then
  \[
      \L_1=\L_2 \ \Longleftrightarrow \ \Shuf(\L_1)\cong \Shuf(\L_2)\ .
  \]
  If $\Shuf(\L_1)$ contains an interval isomorphic to
  $\omega\cdot(\mathbf{i+1})$, then there is a linear order $K$ with
  $\omega\cdot(\mathbf{i+1})+K\in\L_1$.
\end{lemma}

\begin{proof} If $\L_1=\L_2$, then it is clear that $\Shuf(\L_1)\cong
  \Shuf(\L_2)$. Conversely, suppose $f$ is an isomorphism from
  $\Shuf(\L_1)$ to $\Shuf(\L_2)$. We prove that $\L_1=\L_2$. By
  symmetry we only need to prove that $\L_1\subseteq \L_2$.

  Say $\L_j=\{L_{j,s}\mid s\in \N \}$ for $j\in \{1,2\}$. Intuitively,
  for $j\in \{1,2\}$, $\Shuf(\L_j)$ can be viewed as obtained by
  replacing each $q\in \Q$ with a linear order $L(j,q)\cong
  L_{j,c(q)}$, where $c$ is a dense $\N$-coloring.  Fix $q\in
  \Q$. Suppose $f(L(1,q))$ contains elements in $L(2,p)$ and $L(2,p')$
  for $p,p'\in \Q$ with $p<p'$. Then in $f(L(1,q))$ there are
  infinitely many disjoint intervals that are isomorphic to $\omega\cdot
  \mathbf{i}$, while in $L(1,q)$ there is exactly one such interval, a
  contradiction. Therefore $f$ maps $L(1,q)$ into $L(2,p)$ for some
  $p\in \Q$.

  If $f(L(1,q)) \varsubsetneq L(2,p)$, then $f^{-1}(L(2,p))$ contains
  an element $x\notin L(1,q)$.  The argument from the previous
  paragraph with $f$ replaced by $f^{-1}$ again leads to a
  contradiction.  Therefore $f(L(1,q))=L(2,p)$. This means that for
  all $L\in \L_1$, there is $L'\in \L_2$ such that $L\cong L'$ and the
  lemma is proved.

  Let $I\cong\omega\cdot(\mathbf{i+1})$ be some interval in
  $\Shuf(\L_1)$. First suppose there are $p<r$ in $\Q$ such that $I$
  intersects $L(1,p)$ and $L(1,r)$. But then $L(1,q)\subseteq I$ for
  all $p<q<r$, implying that $(\Q,\le)$ embeds into
  $I\cong\omega\cdot(\mathbf{i+1})$ which is impossible.  Hence there
  is some $q\in\Q$ with $I\subseteq L(1,q)\in\L_1$. Then there is a
  linear order $K$ such that $L(1,q)=\omega\cdot\mathbf{i}+K$. Since
  $\omega\cdot\mathbf{i}$ (let alone $\omega\cdot(\mathbf{i+1})$) is
  no interval in $K$, the interval $I$ has to intersect the initial
  segment $\omega\cdot\mathbf{i}$ of $L(1,q)$. But then $\omega$ has
  to be an initial segment of $K$, i.e.,
  $L(1,q)=\omega\cdot(\mathbf{i+1})+K'$ for some linear order
  $K'$.\qed
\end{proof}
Now notice that $\omega\cdot \mathbf{(i+1)}$ is not isomorphic to any interval of 
$L^{i+1}_{\overline{c}}$, $K^{i+1}$, or $M^{i+1}$ (each of the orders $L^i_{\overline{c}x}$, $K^i$, and $M \in \M^i$ is a shuffle sum
and therefore does not start with $\omega$). 
Hence (P3) holds for $i+1$. 
Furthermore, the following holds:
\begin{eqnarray*}
P_{i+1}(\overline{c}) \ & \Longleftrightarrow & \ \exists x \in \N_+ : P_{i+1}(\overline{c},x) \\
 & \Longleftrightarrow & \ \exists x \in \N_+ : L^i_{\overline{c}x} \cong K^i \\
 & \Longleftrightarrow & \  \L^{i+1}_3\subseteq \L^{i+1}_1(\overline{c}) \\
 & \Longleftrightarrow & \  \L^{i+1}_1(\overline{c})\cup \L^{i+1}_2 = \L^{i+1}_2\cup \L^{i+1}_3 \\
 & \stackrel{\text{Lemma~\ref{lem:shuf_inf}}}{\Longleftrightarrow} & \ L^{i+1}_{\overline{c}}\cong K^{i+1} \\
\neg P_{i+1}(\overline{c}) \ & \Longleftrightarrow & \ \forall x \in \N_+ : \neg P_{i+1}(\overline{c},x) \\
 & \Longleftrightarrow & \ \forall x \in \N_+  \; \exists M \in \M^i : L^i_{\overline{c}x} \cong M \\
 & \Longleftrightarrow & \  \L^{i+1}_1(\overline{c})\cup \L^{i+1}_2 = \L^{i+1}_2 \\
 & \stackrel{\text{Lemma~\ref{lem:shuf_inf}}}{\Longleftrightarrow} & \ L^{i+1}_{\overline{c}}\cong M^{i+1}  
 \end{eqnarray*}
 We have shown (P1) and (P2) for $i+1$ in case $i$ is odd.

\subsubsection{Second induction step: from $P_i$ to $P_{i+1}$ for $i$ even} 
\label{second-induction-step}

Let $i \geq 1$ be even and
consider the $\Pi^0_{i+1}$-predicate $P_{i+1}$. 
Again, we write $k$ for $n-i$. For all 
$\overline{c}\in \N^{k}_+$, $P_{i+1}(\overline{c})$ is 
logically equivalent to $\forall x: P_i(\overline{c},x)$. 
Since $i$ is even, we must have $i \geq 2$. Therefore 
the set $\M^i$ is a singleton, consisting of the linear order $M^i$.

Fix $\overline{c}\in \N^{k}_+$. Define the classes of 
linear orders $\L^{i+1}_1(\overline{c})$, $\L^{i+1}_2$, and $\L^{i+1}_3$ using the same definition as in (\ref{eqt:lo_L_i+1}). 
The linear orders $L^{i+1}_{\overline{c}}$, $K^{i+1}$, and $M^{i+1}$ are defined as follows:
\[
    L^{i+1}_{\overline{c}} = \Shuf(\L^{i+1}_1(\overline{c}) \cup \L_3^{i+1}), 
\qquad K^{i+1} = \Shuf( \L_3^{i+1}), \qquad M^{i+1} = \Shuf(\L_2^{i+1}\cup
\L_3^{i+1}) .
\]
Again, $\omega\cdot \mathbf{(i+1)}$ is not isomorphic to any interval of $L^{i+1}_{\overline{c}}$, $K^{i+1}$, or $M^{i+1}$. 
Hence (P3) holds for $i+1$. Furthermore, the following holds:
\begin{eqnarray*}
P_{i+1}(\overline{c}) \ & \Longleftrightarrow & \ \forall x \in \N_+ : P_i(\overline{c},x) \\
 & \Longleftrightarrow & \ \forall x \in \N_+ : L^i_{\overline{c}x} \cong K^i \\
 & \Longleftrightarrow & \  \L^{i+1}_1(\overline{c})\cup \L^{i+1}_3 = \L^{i+1}_3 \\
 & \stackrel{\text{Lemma~\ref{lem:shuf_inf}}}{\Longleftrightarrow} & \ L^{i+1}_{\overline{c}}\cong K^{i+1} \\
\neg P_{i+1}(\overline{c}) \ & \Longleftrightarrow & \ \exists x \in \N_+ : \neg P_i(\overline{c},x) \\
 & \Longleftrightarrow & \ \exists x \in \N_+ : L^i_{\overline{c}x} \cong M^i \\
 & \Longleftrightarrow & \  \L^{i+1}_1(\overline{c})\cup \L^{i+1}_3 = \L^{i+1}_2\cup \L^{i+1}_3 \\
 & \stackrel{\text{Lemma~\ref{lem:shuf_inf}}}{\Longleftrightarrow} & \ L^{i+1}_{\overline{c}}\cong M^{i+1}
\end{eqnarray*}
We have shown (P1) and (P2) for $i+1$ in case $i$ is even.
This finishes the construction and proof for (P1), (P2), and (P3) in the inductive step.

\subsection{Automaticity}

To construct automatic presentations of the linear orders from the
previous section, we first fix some notations.
For $\overline{c} = (c_1,\ldots,c_k) \in \N_+^k$ and a symbol $a$, we re-define $a^{\overline{c}}$ as the word
\[
    a^{c_1}\sharp \cdots a^{c_k} \sharp \in \{a,\sharp\}^*.
\]
Recall that Lemma~\ref{lm:equiv-runs} described a way to represent a polynomial 
$p(\overline{x})\in \N[\overline{x}]$ in $k$ variables using the number of accepting runs of an 
automaton $\A[p(\overline{x})]$. The next lemma re-states Lemma~\ref{lm:equiv-runs} 
with respect to the new definition of $a^{\overline{c}}$.

\begin{lemma}\label{lem:lo_runs}
  {}From a polynomial $p(\overline{x})\in \N[\overline{x}]$ in $k$
  variables, one can effectively construct a non-deterministic
  automaton $\A[p(\overline{x})]$ on alphabet $\{a,\sharp\}$ such that
  $L(\A[p(\overline{x})]) = (a^+\sharp)^k$ and for all
  $\overline{c}\in \N^k_+:$ $\A[p(\overline{x})]$ has exactly
  $p(\overline{c})$ accepting runs on input $a^{\overline{c}}$.
\end{lemma}

\begin{proof}
We use the same proof as for Lemma~\ref{lm:equiv-runs}. The only difference is when the polynomial 
$p(x_1,\ldots,x_k)$ is of the form $x_i$ for some $i\in \{1,\ldots,k\}$.
In this case, the automaton $\A[x_i]$ is $(S,I,\Delta,F)$ where 
$S=\{q_0,q_1,\ldots q_k,q'_i\}$, $I=\{q_0\}$, $F=\{q_k\}$ and the transition relation $\Delta$ is
\[
    \Delta =  \{(q_{j-1},\sharp,q_j)\mid 1 \leq j \leq k, j\neq i\} \cup
             \{(q,a,q)\mid q\in S\} \cup \{(q_{i-1}, a, q'_i), (q'_i,\sharp,q_{i})\} .
\]
It is easy to see that $L(\A[x_i]) = (a^+\sharp)^k$ and $\A[x_i]$ has exactly $c_i$ 
accepting runs on input $a^{\overline{c}}$ where $\overline{c}\in \N^k_+$.\qed
\end{proof}
{}From now on, when referring to $\A[p(\overline{x})]$, we always
assume it is defined in the sense of Lemma~\ref{lem:lo_runs} (as
opposed to Lemma~\ref{lm:equiv-runs}). Let $\A$ be a
non-deterministic finite automaton over the alphabet $\Sigma$ and let
$\Delta$ be the transition relation of $\A$.  Recall the definition of
the automaton $\Run_\A$ and the projection morphism $\pi : \Delta^*
\to \Sigma^*$ from Section~\ref{sec:equiv}.  Then, $\Run_\A$ is an
automaton over the alphabet $\Delta$.  Assume that a lexicographic
order $\leq_{\lex}$ has been defined on each of $\Sigma^*$ and
$\Delta^*$. Define the automatic linear order $\sqsubseteq$ on
$L(\Run_\A)$ such that for all $w,w'\in L(\Run_\A)$:
\begin{equation}\label{eqt:lo_sqsubseteq}
   w\sqsubseteq w' \ \Longleftrightarrow \ \pi(w) <_{\lex} \pi(w') \vee (\pi(w)=\pi(w') \wedge w\leq_{\lex} w').
\end{equation}
Let $\Sigma_i$ be the alphabet $\{\sharp, \$_1,\ldots, \$_{i-1}, \$, 0,1,a,b_1,b_2,b_3\}$. 
Fix the order $<$ on $\Sigma_i$ such that
\begin{equation} \label{def-alph-order}
   \$ < \$_1 < \cdots < \$_{i-1} < 0 < \sharp < a < b_1 < b_2 < b_3 < 1.
\end{equation} 
For any automaton $\A$ over $\Sigma_i$, fix an arbitrary order on the transition relation
$\Delta$ of $\A$. Let $\leq_{\lex}$ be the lexicographic orders on $\Sigma_i^*$ and $\Delta^*$ defined with respect to these orders, 
respectively. From now on, we will always let 
$\sqsubseteq$ be the linear order as defined in (\ref{eqt:lo_sqsubseteq}) with respect to $\leq_{\lex}$.
For a regular language $L \subseteq \Sigma^*$ let $\first(L) = \{ a \in
\Sigma \mid \exists w \in \Sigma^* : aw \in L \}$. 
For $u\in \Sigma^*$, we use $L[u]$ to
denote the language $u\Sigma^* \cap L$. 
Technically, in this section we prove by induction on $i$ the following statement:

\begin{proposition}\label{prop:lo_automaticity} We can compute automata $\A^i$ 
over $\Sigma_i$ such that:
\begin{enumerate}[(1)]
\item $L(\A^1) = ( (a^+\sharp)^n \cup b_1^+\sharp \cup b_2\sharp)\$ R$ for
  some regular language $R \subseteq \Sigma_1^+$
\item  If $i>1$, then $L(\A^i) = ( (a^+\sharp)^{n-i+1} \cup b_1 \sharp \cup b_2\sharp)\$ R$ 
for some regular language $R \subseteq \Sigma_i^+$
\item $L^i_{\overline{c}}\cong (\pi^{-1}(L(\A^i)[\overline{a}^{\overline{c}}])
     \cap L(\Run_{\A^i}); \sqsubseteq)$ for $\overline{c}\in \N^{n-i+1}_+$
\item $M^1_m\cong (\pi^{-1}(L(\A^1)[b_1^m \sharp])\cap L(\Run_{\A^1});\sqsubseteq)$ for $m\in \N_+$
\item  $M^i \cong (\pi^{-1}(L(\A^i)[b_1\sharp]) \cap L(\Run_{\A^i});\sqsubseteq)$ for $i>1$
\item $K^i \cong (\pi^{-1}(L(\A^i)[b_2\sharp])\cap L(\Run_{\A^i});\sqsubseteq)$
\end{enumerate}
Moreover, in (1) and (2) we have $\first(R) \subseteq \{0,1\}$. 
\end{proposition}

\subsubsection{Effective automaticity of shuffle sums} 

This section shows that we can construct an automatic presentation of
the shuffle sum of a class of automatic linear orders that are
presented in some specific way. For a regular language $D$ over an
alphabet, which does neither contain $0$ nor $1$, let $\sigma(D) =
(\{0,1\}^* 1 D)^+$.

\begin{lemma}\label{lem:lo_aut_shuf}
Let $\A$ be an automaton such that $L(\A) = E D \$ F$ for
regular languages $E,D \subseteq \{a,b_1,b_2,b_3,\sharp\}^*$ and 
$F \subseteq \Sigma^*_i$ (for some $1\leq i\leq n$).
We can effectively compute an automaton 
$\sigma(\A, E)$ such that $L(\sigma(\A,E)) = E\$ \sigma(D) \$ F$ and for all $u\in E$:
\[
    (\pi^{-1}(u\$ \sigma(D)\$ F) \cap L(\Run_{\sigma(\A,E)}) ; \sqsubseteq)
    \cong \Shuf(\{(\pi^{-1}(u v \$ F) \cap L(\Run_{\A}); \sqsubseteq) \mid v\in D\}).
\]
\end{lemma}
\begin{proof}
Suppose $\A=(S,I,\Delta, S_f)$.  Let $\Gamma = \{a,b_1,b_2,b_3,\sharp\}$.
We first define the automaton 
\[
\A' = (S\times\{1,2,\Loop\}, I\times \{1\},\Delta', S_f\times \{2\}).
\]
The transition function $\Delta'$ of $\A'$ is defined as follows:
\begin{eqnarray*}
    \Delta'  &=& \{((q,1),\alpha,(p,1)) \mid (q,\alpha,p)\in \Delta, \alpha\in\Gamma \} \cup \\
             & & \{((q,1),\$, (q,\Loop)) \mid q \in S  \} \cup \\
             & & \{((q,\Loop), \alpha, (q,\Loop)) \mid \alpha\in \Gamma\cup \{0,1\} \} \cup \\
             & & \{((q,\Loop), 1, (q,2)) \mid q\in S\} \cup \\
             & & \{((q,2),\alpha,(p,2)) \mid (q,\alpha,p)\in \Delta \}
\end{eqnarray*}
Intuitively, $\A'$ consists of two copies of $\A$ whose state spaces
are $S\times \{1\}$ and $S\times \{2\}$.  The automaton $\A'$ runs by
starting simulating $\A$ on the first copy. When the first $\$$ is
read, it stops the simulation. For this, the automaton stores the
state $q$ by moving to the ``looping state'' $(q, \Loop)$. The
automaton will stay in $(q,\Loop)$ unless 1 is read, in which case, it
may ``guess'' that it reads the last $1$ before the second $\$$ in the
input.  If so, it goes out of $(q,\Loop)$ and continues the simulation
in the second copy of $\A$ and accepts the input word if the run stops
at a final state.  If the guess was not correct and there is another
$1$ before the second $\$$ in the input, then the run will necessarily
reject.
 
It is easy to see that for all $u_1,u_2\in \Gamma^*$, 
$v\in (\Gamma\cup\{0,1\})^* 1$ and 
$u_3\in F$, the number of accepting runs of $\A'$ on 
$u_1\$ v u_2 \$ u_3$ is the same as the number of accepting runs of $\A$ on $u_1 u_2 \$ u_3$, i.e.,
\begin{equation}\label{eqt:lo_same run}
    |L(\Run_{\A'}) \cap \pi^{-1}(u_1\$v u_2 \$ u_3) |  
      = |L(\Run_{\A}) \cap \pi^{-1}(u_1 u_2\$ u_3)|.
\end{equation}
Let 
\[
\sigma(\A,E) \ = \ E\$ \sigma(D) \$ F \ \cap \ \A'.\]
Note that $L(\sigma(\A,E))=E\$ \sigma(D) \$ F$. Also, for any $u_1\in E$, 
$v\in (\{0,1\}^* 1 D)^* \{0,1\}^* 1$, $u_2\in D$, and $u_3\in F$, the number of accepting runs of 
$\sigma(\A,E)$ on $u_1\$ vu_2\$ u_3$ equals the number of accepting runs of $\A'$ on 
$u_1\$ v u_2 \$ u_3$, which is, by (\ref{eqt:lo_same run}), equal to the number of accepting runs of $\A$ on $u_1 u_2\$ u_3$.
Hence, we have
\begin{equation}\label{eqt:lo_same run2}
    |L(\Run_{\sigma(\A,E)}) \cap \pi^{-1}(u_1\$vu_2 \$ u_3) |  
        = |L(\Run_{\A})\cap \pi^{-1}(u_1 u_2\$ u_3)|.
\end{equation}
We prove the following claim.

\medskip

\noindent {\em Claim 1.}  For all $u_1\in E$, $v\in (\{0,1\}^* 1 D)^*\{0,1\}^*1$ and $u_2\in D$, 
\begin{equation}\label{eqt:lo_shuf}
    (\pi^{-1}(u_1\$ v u_2 \$ F)\cap L(\Run_{\sigma(\A,E)}) ; \sqsubseteq) \cong (\pi^{-1}(u_1 u_2 \$ F)\cap L(\Run_\A) ; \sqsubseteq).
\end{equation}
For $u\in F$, let $L(u)=(\pi^{-1}(u_1 u_2 \$ u)\cap L(\Run_\A); \sqsubseteq)$. 
Note that this is a finite linear order.
Consider the linear order $(F;\leq_{\lex})$.
By definition of $\sqsubseteq$,
\[
    (\pi^{-1}(u_1 u_2 \$ F)\cap L(\Run_\A) ; \sqsubseteq) \cong  \sum_{u\in F} L(u).
\]
By (\ref{eqt:lo_same run2}), $L(u)\cong  (\pi^{-1}(u_1\$ v u_2 \$ u)\cap L(\Run_{\sigma(\A,E)}) ; \sqsubseteq)$. By definition of $\sqsubseteq$ again,
\begin{align*}
    (\pi^{-1}(u_1\$ v u_2 \$ F)\cap L(\Run_{\sigma(\A,E)}) ; \sqsubseteq)  
   & \cong \sum_{u\in F} (\pi^{-1}(u_1\$ v u_2 \$ u)\cap L(\Run_{\sigma(\A,E)}) ; \sqsubseteq) \\ 
   &\cong  \sum_{u\in F} L(u) \\ 
   &\cong (\pi^{-1}(u_1 u_2 \$ F)\cap L(\Run_\A) ; \sqsubseteq) .
\end{align*}
This proves Claim~1.

\medskip

\noindent
Let $c:\sigma(D)\rightarrow D$ be the function such that
\[
    \forall x\in (\{0,1\}^*1D)^*\{0,1\}^* 1\ \forall u\in D:\ c(xu)=u.
\] 
{\em Claim 2.} $(\sigma(D);\leq_{\lex})\cong (\Q;\leq)$ and the function $c$ is a dense $D$-coloring of $(\sigma(D);\leq_{\lex})$.

\medskip

\noindent
First, for every $w=x1u\in \sigma(D)$ with $x\in (\{0,1\}^*1D)\{0,1\}^*$ and $u\in D$, we have
\[
    x01u <_{\lex} w <_{\lex} x11u.
\]
Hence, $(\sigma(D); \leq_{\lex})$ does not have a smallest or largest element. 
It remains to show that the linear order $(\sigma(D); \leq_{\lex})$ is densely 
$D$-colored by $c$ (this implies that $(\sigma(D); \leq_{\lex})$ is dense and 
hence, by Cantor's theorem, isomorphic to $(\Q;\leq)$). Consider two words 
$w_1,w_2\in \sigma(D)$ such that $w_1<_{\lex} w_2$. There are two cases.

\medskip

\noindent Case 1. $w_1 = x\alpha y$, $w_2=x\beta z$ for 
$x,y,z\in (\Gamma \cup\{0,1\})^*$ and $\alpha,\beta\in \Gamma\cup\{0,1\}$ 
such that $\alpha<\beta$. In this case, for all $u\in D$, we have
\[
    w_1 <_{\lex} w_1 1 u <_{\lex} w_2 \quad\text{and}\quad w_1 1 u \in
    \sigma(D) .
\]

\noindent Case 2. $w_2=w_1 x$ for some $x\in (\Gamma \cup\{0,1\})^+$. Since 
$w_2\in \sigma(D)$, we have $x\notin 0^*$. Say $x=0^j\alpha y$ for some $j\geq 0$, 
$\alpha \neq 0$ and $y\in (\Gamma \cup\{0,1\})^*$. 
We must have $\alpha \in \{1, a,b_1,b_2,b_3,\sharp\}$. Since every symbol
from this set is larger than $0$ (see (\ref{def-alph-order})) we must have $\alpha > 0$.
Then for all $u\in D$, we have
\[
    w_1 <_{\lex} w_1 0^{j+1} 1 u <_{\lex} w_2 \quad\text{and}\quad w_1 0^{j+1} 1 u \in
    \sigma(D).
\]
Hence $(\sigma(D);\leq_{\lex})$ is indeed densely colored by $c$. This proves Claim~2.

\medskip

\noindent
Since $\$$ is the minimum in the order $<$ on $\Sigma_i$, for any $u\in E$, $v,v'\in \sigma(D)$ and $w,w'\in F$, we have
\[
   v <_{\lex} v' \ \Longrightarrow \ u\$ v \$ w <_{\lex} u \$ v'\$ w'.
\]
Therefore,
\begin{eqnarray*}
    (\pi^{-1}(u\$ \sigma(D) \$ F) \cap L(\Run_{\sigma(\A,E)}) ; \sqsubseteq) 
     &\cong & \sum_{v\in \sigma(D)}(\pi^{-1}(u\$ v \$ F) \cap L(\Run_{\sigma(\A,E)}); \sqsubseteq) \\
     & \stackrel{\text{Claim 1}}{\cong} & \sum_{v\in \sigma(D)}(\pi^{-1}(u c(v) \$ F) \cap L(\Run_{\A}); \sqsubseteq) \\
     & \stackrel{\text{Claim 2}}{\cong} & \Shuf(\{(\pi^{-1}(u v\$ F)\cap L(\Run_\A);  \sqsubseteq) \mid v\in D\}) .
\end{eqnarray*}
\qed
\end{proof}

\subsubsection{Base case: automatic presentations for  $L^1_{\overline{c}},
  K^1$, and $M^1_m$}

Recall the notations from Section~\ref{construction-lo-ind-base}.
In the following, if $D$ is a regular language and
$\A$ is a finite non-deterministic automaton then we denote by $D \A$ a finite
automaton that results from the disjoint union of 
a deterministic automaton $\A_D$ for $D$ and the automaton 
$\A$ by adding all transitions $(q,a,p)$ where:
(i) $q$ is a state of $\A_D$, (ii) there is a transition 
$(q,a,q')$ in $\A_D$, where $q'$ is a final state of $\A_D$, 
and (iii) $p$ is an initial state of $\A$. 
Clearly, $L (D \A) = D L(\A)$. We will only apply this
definition in case the product $D L(A)$ is unambiguous.
This means that if $u \in D L(A)$ then there exists a unique
factorization $u = u_1 u_2$ with $u_1 \in D$ and $u_2 \in L(A)$.
The following lemma is easy to prove:

\begin{lemma} \label{lemma-unambiguous}
Let $\A$ be a finite non-deterministic automaton
and let $D$ be a regular language such that the 
product $D L(A)$ is unambiguous. Let $u_1 \in D$ 
and $u_2 \in L(\A)$. Then, the
number of accepting runs of $D \A$ on $u_1u_2$ 
equals the number of accepting runs of $\A$ on $u_2$.
\end{lemma}


\begin{lemma}\label{lem:lo_poly}
{}From two given polynomials $q_1(\overline{x}),q_2(\overline{x})\in
\N[\overline{x}]$ in $k$ variables, one can effectively construct 
an automaton $\A[q_1,q_2]$ over the alphabet $\{a,\#,\$\}$ such that
\begin{itemize}
\item $L(\A[q_1,q_2]) = (a^+\sharp)^k \$$ and
\item For all $\overline{c}\in \N_+^k$,  
 $(\pi^{-1}(a^{\overline{c}}\$) \cap L(\Run_{\A[q_1,q_2]}); \sqsubseteq) \cong L[q_1(\overline{c}), q_2(\overline{c})]$.
\end{itemize}
\end{lemma}
\begin{proof}
We construct $\A[q_1,q_2]$ by taking a copy of
$\A[C(q_1(\overline{x}),q_2(\overline{x}))]$ (see Lemma~\ref{lem:lo_runs}), 
adding a new state $q_{\$}$ and transitions $(q_f, \$, q_{\$})$ for each accepting state 
$q_f$ in $\A[C(q_1(\overline{x}),q_2(\overline{x}))]$ and making $q_{\$}$ the only accepting state of $\A[q_1,q_2]$.
Note that for any $\overline{c}\in \N_+^k$, the number of accepting runs of $\A[q_1,q_2]$ 
on $a^{\overline{c}}\$$ is the same as the number of accepting 
runs of $\A[C(q_1(\overline{x}),q_2(\overline{x}))]$ on $a^{\overline{c}}$, which is equal to $C(q_1(\overline{c}), q_2(\overline{c}))$.
Hence, $(\pi^{-1}(a^{\overline{c}}\$) \cap L(\Run_{\A[q_1,q_2]}); \sqsubseteq)$ 
forms a copy of $L[q_1(\overline{c}), q_2(\overline{c})]$ and the lemma is proved.
\qed
\end{proof}
By Lemma~\ref{lem:lo_poly}, we can construct automata
$\A_1=\A[p_1(\overline{x})+x_{\ell+1}, p_2(\overline{x})+x_{\ell+1}]$,
where $\overline{x}\in \N_+^\ell$, over the alphabet
$\{a,\sharp,\$\}$, $\A_2=\A[x_1+x_2,x_1+x_2]$ over the alphabet $\{b_1,
\sharp, \$\}$, $\A_3=\A[x_1+x_2,x_1]$ over the alphabet $\{b_2,\sharp, \$\}$
and $\A_4=\A[x_1,x_1+x_2]$ over the alphabet $\{b_3,\sharp, \$\}$ such that:
\begin{eqnarray}
\forall \overline{c}\in \N_+^\ell \ \forall c_{\ell+1}\in \N_+ :
(\pi^{-1}(a^{\overline{c}c_{\ell+1}}\$) \cap L(\Run_{\A_1});\sqsubseteq)
 & \cong & L[p_1(\overline{c})+c_{\ell+1}, p_2(\overline{c})+c_{\ell+1}] \label{eq-A1} \\
\forall e_1,e_2\in \N_+ :
  (\pi^{-1}(b_1^{e_1}\sharp b_1^{e_2}\sharp\$) \cap L(\Run_{\A_2});
  \sqsubseteq) & \cong & L[e_1+e_2,e_1+e_2] \label{eq-A2} \\
\forall e_1,e_2\in \N_+ : 
 (\pi^{-1}(b_2^{e_1}\sharp b_2^{e_2}\sharp\$) \cap L(\Run_{\A_3});
 \sqsubseteq) & \cong & L[e_1+e_2,e_1]  \label{eq-A3} \\
\forall e_1,e_2\in \N_+ :
(\pi^{-1}(b_3^{e_1}\sharp b_3^{e_2}\sharp\$) \cap L(\Run_{\A_4}); \sqsubseteq)
 & \cong & L[e_1,e_1+e_2]   \label{eq-A4}
\end{eqnarray}
Define the following automata:
\[
    \A^0_1 = \A_1 \uplus ( (a^+\sharp)^n ( \A_3 \uplus \A_4 )), \qquad
    \A^0_2 = \A_2 \uplus (b_1^+ \sharp ( \A_3 \uplus \A_4 )), \qquad
    \A^0_3 = b_2 \sharp ( \A_3 \uplus \A_4).
\]
Note that
\begin{align*}
    L(\A^0_1) &= (a^+\sharp)^n \biggl( (a^+\sharp)^{\ell-n+1} \cup
    (b_2^+\sharp)^2  \cup (b_3^+\sharp)^2 \biggr) \$ , \\
    L(\A^0_2) &= b_1^+ \sharp \biggl( b_1^+ \sharp \cup (b_2^+\sharp)^2 \cup
    (b_3^+\sharp)^2\biggr)  \$ ,\\
    L(\A^0_3) &= b_2 \sharp \biggl((b_2^+\sharp)^2 \cup (b_3^+\sharp)^2 \biggr) \$ .
\end{align*}
Hence, applying Lemma~\ref{lem:lo_aut_shuf} (with $F=\{\varepsilon\}$), 
we can effectively construct automata $\A^1_j$ ($j\in \{1,2,3\}$) as follows:
\[
    \A^1_1 = \sigma( \A^0_1, (a^+\sharp)^n), \qquad \A^1_2 = \sigma(\A^0_2, b_1^+\sharp), \qquad \A^1_3 = \sigma(\A^0_3, b_2\sharp).
\]
For all $\overline{c}\in \N^n_+$ we get:
\begin{align*}
& (\pi^{-1}(L(\A^1_1)[a^{\overline{c}}]) \cap L(\Run_{\A^1_1}); \sqsubseteq) 
\stackrel{\text{Lemma~\ref{lem:lo_aut_shuf}}}{\cong} \\[2mm]
& \qquad \Shuf(\{(\pi^{-1}(a^{\overline{c}} v \$) \cap L(\Run_{\A^0_1}); \sqsubseteq) \mid v\in (a^+\sharp)^{\ell-n+1} \cup
    (b_2^+\sharp)^2  \cup (b_3^+\sharp)^2 \}) = \\[2mm]
& \qquad \Shuf(\{(\pi^{-1}(a^{\overline{c}\,\overline{e}} \$) \cap
                  L(\Run_{\A^0_1}); \sqsubseteq) \mid \overline{e} \in
                  \N_+^{\ell-n+1} \} \; \cup \\
& \qquad \phantom{\Shuf(} \{(\pi^{-1}(a^{\overline{c}} b_2^{e_1}\sharp b_2^{e_2}\sharp \$) \cap
                  L(\Run_{\A^0_1}); \sqsubseteq) \mid e_1,e_2 \in \N_+ \} \; \cup \\
& \qquad \phantom{\Shuf(} \{(\pi^{-1}(a^{\overline{c}} b_3^{e_1}\sharp b_3^{e_2}\sharp \$) \cap
                  L(\Run_{\A^0_1}); \sqsubseteq) \mid e_1,e_2 \in \N_+ \})\stackrel{\text{Lemma~\ref{lemma-unambiguous}}}{\cong}  \\[2mm]
& \qquad \Shuf(\{(\pi^{-1}(a^{\overline{c}\,\overline{e}} \$) \cap
                  L(\Run_{\A_1}); \sqsubseteq) \mid \overline{e} \in
                  \N_+^{\ell-n+1} \} \; \cup \\
& \qquad \phantom{\Shuf(} \{(\pi^{-1}(b_2^{e_1}\sharp b_2^{e_2}\sharp \$) \cap
                  L(\Run_{\A_3}); \sqsubseteq) \mid e_1,e_2 \in \N_+ \} \; \cup \\
& \qquad \phantom{\Shuf(} \{(\pi^{-1}(b_3^{e_1}\sharp b_3^{e_2}\sharp \$) \cap
                  L(\Run_{\A_4}); \sqsubseteq) \mid e_1,e_2 \in \N_+ \})
                  \stackrel{\text{(\ref{eq-A1})--(\ref{eq-A4})}}{=} \\[2mm]
& \qquad \Shuf(\{ L[p_1(\overline{c}, \overline{e})+e_{\ell+1}, 
   p_2(\overline{c},\overline{e})+e_{\ell+1}] \mid \overline{e} \in
                  \N_+^{\ell-n}, e_{\ell+1} \in \N_+ \} \; \cup \\
& \qquad \phantom{\Shuf(} \{ L[e_1+e_2,e_1]  \mid e_1,e_2 \in \N_+ \} \; \cup
\; \{ L[e_1, e_1+e_2]  \mid e_1,e_2 \in \N_+ \}) \stackrel{\text{(\ref{class-L_1})--(\ref{class-L_4})}}{=} \\[2mm]
& \qquad \Shuf( \L^1_1(\overline{c})\cup \L^1_3 \cup \L^1_4) \cong L^1_{\overline{c}}
\end{align*}
Similar calculations yield:
\begin{eqnarray*} 
\forall m\in \N_+ : 
(\pi^{-1}(L(\A^1_2)[b_1^m \sharp]) \cap L(\Run_{\A^1_2}); \sqsubseteq) 
& \cong & \Shuf(\L^1_2(m) \cup \L^1_3 \cup \L^1_4) \cong M^1_m \\ 
(\pi^{-1}(L(\A^1_3)[b_2\sharp])\cap L(\Run_{\A^1_3}); \sqsubseteq) & \cong & \Shuf(\L^1_3 \cup \L^1_4) \cong K^1
\end{eqnarray*}
Let $\A^1 =  \A^1_1\uplus \A^1_2 \uplus \A^1_3$. It is easy to see that  
$L(\A^1) = ( (a^+\sharp)^n \cup b_1^+ \sharp \cup b_2 \sharp) \$ R$ for some 
regular language $R \subseteq \Sigma_1^+$ with $\first(R)
\subseteq \{0,1\}$. Hence $\A^1$ satisfies the statement 
in Proposition~\ref{prop:lo_automaticity}.

\subsubsection{First inductive step: automatic presentations for
  $L^{i+1}_{\overline{c}}$, $K^{i+1}$, $M^{i+1}$ for $i$ odd}

Let $i\geq 1$ be an odd number. 
Recall the notations from Section~\ref{first-induction-step}.
We write $k$ for $n-i$. By applying the inductive
assumption, we obtain an automaton $\A^i$ such that  $L(\A^i) =
((a^+\sharp)^{k+1} \cup \beta\sharp \cup b_2 \sharp) \$ R$ for some regular 
language $R \subseteq \Sigma_i^*$ where $\beta= b_1^+$ if $i=1$, 
and $\beta=b_1$ otherwise. Furthermore, $\first(R)
\subseteq \{0,1\}$ and the following hold for $\A^i$:
\begin{eqnarray}
\forall \overline{c}\in \N^{k+1}_+:\ L^i_{\overline{c}} & \cong &
(\pi^{-1}(a^{\overline{c}} \$ R)\cap L(\Run_{\A^i}); \sqsubseteq) \label{eq-L^i} \\
\M^i &   \cong & \{(\pi^{-1}(u \sharp \$ R) \cap L(\Run_{\A^i});
\sqsubseteq) \mid u \in \beta\} \label{eq-M^i} \\
K^i & \cong & (\pi^{-1}(b_2\sharp\$ R)\cap L(\Run_{\A^i});\sqsubseteq) \label{eq-K^i}
\end{eqnarray}
For any $1\leq j\leq n$, let $S_j = \$_1^+\cup \cdots \cup \$_j^+$.
It is easy to see that
\begin{equation} \label{order on S_j}
    (S_j ;\leq_{\lex}) \cong \omega\cdot \mathbf{j}.
\end{equation}
Define the automata $\B^{i}_1$, $\B^{i}_2$, and $\B^{i}_3$ as 
\begin{eqnarray}
    \B^i_1 &=&  ((a^+\sharp)^{k+1}\$ R \ \cap \ \A^i) \ \uplus \
    (a^+\sharp)^{k+1}\$ S_i, \label{B^i_1} \\
    \B^i_2 &=&  (\beta\sharp\$ R \ \cap \ \A^i) \ \uplus\ \beta\sharp\$ S_i,
    \label{B^i_2} \\
    \B^i_3 &=&  (b_2\sharp\$ R \ \cap \ \A^i) \uplus \ b_2\sharp \$ S_i . \label{B^i_3}
\end{eqnarray}
By (\ref{def-alph-order}), (\ref{eq-L^i})--(\ref{order on S_j}), 
and the fact that $\first(R)\subseteq \{0,1\}$, we have
\begin{eqnarray}
\forall \overline{c}\in \N^{k+1}_+: 
( \pi^{-1}(a^{\overline{c}}\$ (S_i\cup R))\cap L(\Run_{\B^i_1}); \sqsubseteq) 
           & \cong & \omega\cdot \mathbf{i} + L^i_{\overline{c}}, \label{eq-B^i_1} \\
\{(\pi^{-1}(u \sharp\$ (S_i\cup R)) \cap L(\Run_{\B^i_2}); \sqsubseteq)\mid
           u \in \beta\} 
           & \cong & \{\omega\cdot \mathbf{i}+M\mid M\in \M^i\}, \label{eq-B^i_2} \\
    (\pi^{-1}(b_2\sharp\$ (S_i\cup R)) \cap L(\Run_{\B^i_3}); \sqsubseteq) 
           & \cong & \omega\cdot \mathbf{i} + K^i . \label{eq-B^i_3}
\end{eqnarray}
Now construct the automata $\C^{i}_1, \C^{i}_2$, and $\C^{i}_3$ as follows:
\[
\C^{i}_1 =  \B^i_1 \uplus (a^+\sharp)^{k}\B^i_2, \qquad 
\C^{i}_2 = b_1\sharp \B^i_2, \qquad 
\C^{i}_3 = b_2\sharp (\B^i_2\uplus \B^i_3).
\]
We have
\begin{eqnarray*}
    L(\C^{i}_1) & = & (a^+\sharp)^{k} (a^+\sharp \cup \beta\sharp) \$ (S_i\cup R), \\
    L(\C^{i}_2) & = & b_1 \sharp \beta \sharp \$ (S_i\cup R),\\
    L(\C^{i}_3) & = & b_2 \sharp (\beta \sharp \cup b_2 \sharp) \$ (S_i\cup R).
\end{eqnarray*}
Hence, we can apply Lemma~\ref{lem:lo_aut_shuf} to $\C^{i}_1$, $\C^{i}_2$, and
$\C^{i}_3$ (with $F = S_i \cup R$) to define the following automata:
\[
  \A^{i+1}_1 =  \sigma(\C^i_1, (a^+\sharp)^{k}), \qquad 
  \A^{i+1}_2 = \sigma(\C^i_2,b_1\sharp), \qquad \A^{i+1}_3 = \sigma(\C^i_3, b_2\sharp).
\]
For all $\overline{c}\in \N^k_+$ we get:
\begin{align*}
& (\pi^{-1}(L(\A^{i+1}_1)[a^{\overline{c}}]) \cap L(\Run_{\A^{i+1}_1}); \sqsubseteq) 
\stackrel{\text{Lemma~\ref{lem:lo_aut_shuf}}}{\cong} \\[2mm]
& \qquad \Shuf(\{(\pi^{-1}(a^{\overline{c}} v \$ (S_i \cup R)) \cap
L(\Run_{\C^i_1}); \sqsubseteq) \mid v\in a^+\sharp \cup \beta \sharp \}) = \\[2mm]
& \qquad \Shuf(\{(\pi^{-1}(a^{\overline{c}\,e} \$ (S_i \cup R)) \cap
                  L(\Run_{\C^i_1}); \sqsubseteq) \mid e \in \N_+ \} \; \cup \\
& \qquad \phantom{\Shuf(} \{(\pi^{-1}(a^{\overline{c}} u \sharp \$ (S_i \cup R)) \cap
                  L(\Run_{\C^i_1}); \sqsubseteq) \mid u \in \beta \}) \stackrel{\text{Lemma~\ref{lemma-unambiguous}}}{\cong}  \\[2mm]
& \qquad \Shuf(\{(\pi^{-1}(a^{\overline{c}\,e} \$ (S_i \cup R)) \cap
                  L(\Run_{\B^i_1}); \sqsubseteq) \mid e \in \N_+ \} \; \cup \\
& \qquad \phantom{\Shuf(} \{(\pi^{-1}(u \sharp \$ (S_i \cup R)) \cap
                  L(\Run_{\B^i_2}); \sqsubseteq) \mid u \in \beta \})
                           \stackrel{\text{(\ref{eq-B^i_1}), (\ref{eq-B^i_2})}}{=}
                           \\[2mm]
& \qquad \Shuf(\{\omega\cdot \mathbf{i}+\L^i_{\overline{c} e} \mid e\in \N_+\} \cup
    \{\omega\cdot \mathbf{i}+M \mid M\in \M^i\}) 
   \stackrel{\text{(\ref{eqt:lo_L_i+1}), (\ref{order-for-i+1})}}{\cong}  L^{i+1}_{\overline{c}} 
\end{align*}
Similarly, we can show:
\begin{eqnarray*}
(\pi^{-1}(L(\A^{i+1}_2)[b_1\sharp]) \cap L(\Run_{\A^{i+1}_2}); \sqsubseteq)
  & \cong & \Shuf(\{\omega\cdot \mathbf{i}+M\mid M\in \M^i\}) \cong  M^{i+1}, \\
(\pi^{-1}(L(\A^{i+1}_3)[b_2\sharp])\cap L(\Run_{\A^{i+1}_3}); \sqsubseteq) 
 & \cong & \Shuf(\{\omega\cdot \mathbf{i}+M\mid M\in \M^i\} \cup \{\omega\cdot
 \mathbf{i}+K^i\}) \cong  K^{i+1}.
\end{eqnarray*}
Let $\A^{i+1} =  \A^{i+1}_1\uplus \A^{i+1}_2 \uplus \A^{i+1}_3$. 
It is easy to see that  $L(\A^{i+1}) = ((a^+\sharp)^{k} \cup
b_1\sharp \cup b_2\sharp) \$ R'$ 
for some regular language $R' \subseteq \Sigma_{i+1}^+$ with $\first(R')
\subseteq \{0,1\}$. 
Hence $\A^{i+1}$ satisfies the statement in Proposition~\ref{prop:lo_automaticity}.

\subsubsection{Second inductive step: automatic presentations for
  $L^{i+1}_{\overline{c}}$, $K^{i+1}$, $M^{i+1}$ for $i$ even}

Using the same technique,  we can construct 
automatic presentations for  $L^{i+1}_{\overline{c}}$ ($\overline{c}\in
\N^{k}_+$), $M^{i+1}$, and $K^{i+1}$ in case $i$ is even.
We first define the automata $\B^{i}_1$, $\B^{i}_2$, and $\B^{i}_3$
as in (\ref{B^i_1})--(\ref{B^i_3}), with $\beta = b_1$ this time.
Then we construct
\[
    \C^{i}_1 =  \B^{i}_1 \uplus (a^+\sharp)^{k} \B^{i}_3, \qquad 
    \C^{i}_2 = b_1\sharp (\B^{i}_2 \uplus \B^{i}_3), \qquad 
    \C^{i}_3 = b_2\sharp \B^i_3.
\]
We define the following automata by applying Lemma~\ref{lem:lo_aut_shuf}:
\[
    \A^{i+1}_1 =  \sigma(\C^{i}_1, (a^+\sharp)^{k}), \qquad 
    \A^{i+1}_2 = \sigma(\C^{i}_2,b_1\sharp), \qquad 
    \A^{i+1}_3 = \sigma(\C^{i}_3,b_2\sharp).
\]
By Lemma~\ref{lem:lo_aut_shuf}, it is easy to check the following:
\begin{eqnarray*}
\forall \overline{c}\in \N^{k}_+: (\pi^{-1}(L(\A^{i+1}_1)[a^{\overline{c}}])
\cap L(\Run_{\A^{i+1}_1}); \sqsubseteq) & \cong & 
\Shuf(\{\omega\cdot \mathbf{i}+\L^{i}_{\overline{c}x} \mid x\in \N_+\} \cup
\{\omega\cdot \mathbf{i} + K^{i}\})\\
& \cong & L^{i+1}_{\overline{c}}, \\
(\pi^{-1}(L(\A^{i+1}_2)[b_1\sharp]) \cap L(\Run_{\A^{i+1}_2}); \sqsubseteq) &
\cong & \Shuf(\{\omega\cdot \mathbf{i}+M^{i}\} \cup \{\omega\cdot \mathbf{i}+K^{i}\} ) \\
& \cong & M^{i+1}, \\
(\pi^{-1}(L(\A^{i+1}_3)[b_2\sharp]) \cap L(\Run_{\A^{i+1}_3}); \sqsubseteq) &
\cong & \Shuf(\{\omega\cdot \mathbf{i}+K^{i}\}) \\
& \cong & K^{i+1} .
\end{eqnarray*}
Let $\A^{i+1}=\A^{i+1}_1\uplus \A^{i+1}_2\uplus \A^{i+1}_3$. 
It is easy to see that $L(\A^{i+1})\subseteq ((a^+\sharp)^{k}\cup b_1\sharp
\cup b_2\sharp) \$ R'$ for some regular language $R' \subseteq \Sigma_{i+1}^+$
with $\first(R') \subseteq \{0,1\}$. 
Hence $\A^{i+1}$ satisfies the statement in
Proposition~\ref{prop:lo_automaticity}. 
This finishes the construction in the inductive step and hence 
the proof of Proposition~\ref{prop:lo_automaticity}.   Hence we obtain:

\begin{theorem}\label{thm:lo}
  The isomorphism problem for the class of automatic linear orders is at least as hard as 
  $\FOTh(\N;+,\times)$.
\end{theorem} 
In \cite{KhoRS05}, it is shown that every linear order has finite
FC-rank. We do not define the FC-rank of a linear order in general, see e.g.\ 
\cite{KhoRS05}. A linear order $(L,\leq)$ has FC-rank 1, if after identifying all
$x,y \in L$ such that the interval $[x,y]$ is finite, one obtains a dense
ordering or the singleton linear order.
The result of \cite{KhoRS05} mentioned above 
suggests that the isomorphism problem might be
simpler for linear orders of low FC-rank. We now prove that this is
not the case:

\begin{corollary}
  The isomorphism problem for automatic linear orders of FC-rank 1 is
  at least as hard as $\FOTh(\N;+,\times)$.
\end{corollary}

\begin{proof}
  We provide a reduction from the isomorphism problem for automatic
  linear orders (of arbitrary rank): if $(L,\le)$ is an automatic
  linear order, then so is $(K,\le)= ((-1,0]+[1,2)) \cdot (L,\le)$ (this
  linear order is obtained from $L$ by replacing each point with a
  copy of the rational numbers in $(-1,0]\cup[1,2)$). Then $(K,\le)$
  has FC-rank 1: Only the copies of $0$ and $1$ will be identified,
  and the resulting order is isomorphic to $(\mathbb{Q},\leq)$.
  Moreover, $(L,\le)$ is isomorphic to the set of all $x\in
  K$ satisfying $\exists z > x \; \forall y : (x<y\le z \to
  y=z)$. Hence $(L,\le)\cong(L',\le')$ if and only if
  $((-1,0]+[1,2)) \cdot (L,\le) \cong ((-1,0]+[1,2)) \cdot (L',\le')$, which
  completes the reduction.\qed
\end{proof}

\section{Conclusion}

This paper looks at the isomorphism problem of some typical classes of
automatic structures. Such classes include equivalence structures, successor
trees of height at most $n\in \N$, and linear orders. In particular, we
demonstrate, respectively, $\Pi^0_1$-completeness and $\Pi^0_{2n-3}$-completeness for the
isomorphism problem of the first two classes. The uniformity in our proof
shows that the isomorphism problem of automatic trees of finite height is
recursively equivalent to $\FOTh(\N;+,\times)$.
Similarly, we  prove that the isomorphism problem
of automatic linear orders is  at least as hard as $\FOTh(\N;+,\times)$.
The same technique is also used to proved that the isomorphism problem 
of recursive trees of height at most $n$ is $\Pi^0_{2n}$-complete. 

We conclude with an application of Theorems~\ref{thm:tree} and
\ref{thm:lo}. The following corollary shows that although automatic
structures look simple (especially for automatic trees), there may be
no ``simple'' isomorphism between two automatic copies of the same
structure. An isomorphism $f$ between two automatic structures
with domains $L_1$ and $L_2$, respectively, is 
a $\Sigma^0_k$-isomorphism, if the set $\{ (x,f(x)) \mid x \in L_1\}$
belongs to $\Sigma^0_k$.

\begin{corollary}
  For any $k\in \N$, there
  exist two isomorphic automatic trees of finite height (and two automatic
  linear orders) without any $\Sigma^0_k$-isomorphism.
\end{corollary}
\begin{proof} 
Let $T_1=(D_1;E_1)$ and $T_2=(D_2;E_2)$ be two automatic trees.
Let $P_1(x,y), P_2(x,y),\ldots$ be an effective enumeration of all
binary $\Sigma^0_k$-predicates. This means that from given $e \geq 1$ we can 
effectively compute a description (e.g. a $\Sigma_k$-formula over
$(\N;+,\times)$) of the predicate $P_e(x,y)$.
We define the statement $\text{iso}(T_1,T_2,k)$ as follows:
\begin{align*}
\exists e\ &
\begin{array}[t]{ll}
\forall x_1,x_2\in D_1 \; \exists y_1,y_2\in D_2: \ 
  &  P_e(x_1,y_1)
     \; \wedge \; P_e(x_2,y_2) \; \wedge \\
  &   (x_1 = x_2 \leftrightarrow y_1 = y_2) \; \wedge \;
   ((x_1,x_2)\in E_1 \leftrightarrow (y_1,y_2)\in E_2)
 \end{array}\\
&\land\forall y\in D_2 \; \exists x\in D_1: \ 
    P_e(x,y)
\end{align*}
Since $P_e$ is a $\Sigma^0_k$-predicate, this is a
$\Sigma^0_{k+2}$-statement, which expresses the existence of a
$\Sigma_k^0$-isomorphism from $T_1$ to $T_2$.

By Theorem~\ref{thm:tree}, there is a natural number $n$ such that the
isomorphism problem on the class $\T_n$ of automatic trees of height at most
$n$ is $\Sigma_{k+3}$-hard. 
If for all $T_1,T_2\in \T_n$ with $T_1\cong T_2$ there exists 
a $\Sigma^0_k$-isomorphism from $T_1$ to $T_2$, then the isomorphism problem on $\T_n$ reduces to checking
existence of a $\Sigma^0_k$-isomorphism, which is in $\Sigma^0_{k+2}$ by the
above consideration. Hence, there must be $T_1,T_2\in \T_n$ with 
$T_1\cong T_2$ but there is no $\Sigma^0_k$-isomorphism between them. 

The corollary for linear orders can be proved in 
the same way, where in the definition of 
$\text{iso}(T_1,T_2,k)$ we replace $(x_1,x_2)\in E_1 \leftrightarrow (y_1,y_2)\in E_2$ with 
$x_1 <_1 x_2 \leftrightarrow y_1 <_2 y_2$,
where $<_1$ and $<_2$ are the linear orders of $T_1$ and $T_2$, respectively. 
\qed
\end{proof}


\def\cprime{$'$}

\end{document}